\documentclass[11pt,a4paper]{article}

\usepackage[T1]{fontenc}
\usepackage[utf8]{inputenc}
\usepackage[margin=1in]{geometry}

\usepackage{amsmath,amssymb,amsthm}
\usepackage{mathtools}
\usepackage[linesnumbered,ruled,vlined,noresetcount]{algorithm2e}
\DontPrintSemicolon
\SetArgSty{textnormal}
\SetKwInput{Colors}{Colors}
\SetKwInput{RuleFunction}{Rule Function}
\usepackage{color}
\usepackage{booktabs}
\usepackage{float}
\usepackage{placeins}
\usepackage{enumerate}
\usepackage{tikz}
\usetikzlibrary{arrows.meta,calc}
\usepackage[unicode,pdfencoding=auto,hidelinks]{hyperref}
\usepackage{bookmark}
\theoremstyle{plain}
\newtheorem{theorem}{Theorem}
\newtheorem{jtheorem}[theorem]{Theorem}
\newtheorem{jlemma}[theorem]{Lemma}
\newtheorem{jcorollary}[theorem]{Corollary}

\newcommand{\protocol}[1]{\textsc{#1}}
\newcommand{\colorName}{\mathsf}
\newcommand{\cinit}{\colorName{init}}
\newcommand{\cpath}{\colorName{path}}
\newcommand{\cneigh}{\colorName{neigh}}
\newcommand{\cfin}{\colorName{fin}}
\newcommand{\cheadd}{\colorName{head}}
\newcommand{\cfront}{\colorName{front}}

\newcommand{\clzero}{\colorName{\ell_0}}
\newcommand{\clone}{\colorName{\ell_1}}
\newcommand{\cstay}{\colorName{stay}}
\newcommand{\cstop}{\colorName{stop}}

\newcommand{\mif}{\mathbf{if~}}
\newcommand{\melif}{\mathbf{else~if~}}
\newcommand{\moth}{\mathbf{otherwise}}

\newcommand{\calA}{\mathcal{A}}
\newcommand{\calM}{\mathcal{M}}
\newcommand{\sus}{\protocol{US5}}
\newcommand{\aus}{\calA_{\sus}}
\newcommand{\cus}{C_{\sus}}
\newcommand{\xius}{\xi_{\sus}}
\newcommand{\gktf}{\mathcal{G}_{\mathrm{KTF}}}
\newcommand{\gcb}{\mathcal{G}_{\mathrm{CB}}}
\newcommand{\rlabel}[1]{\textbf{Rule #1}}
\newcommand{\rlabels}[1]{\textbf{Rules #1}}
\newcommand{\itup}[3]{\mathtt{#1#2#3}}
\newcommand{\resulttheorem}[1]{(Thm.~\ref{#1})}
\newcommand{\resultlemma}[1]{(Lem.~\ref{#1})}
\newcommand{\upf}{N(u_k)\setminus
    \left(F\cup\{u_1,u_2,\dots,u_k\}\cup
    \bigcup_{i=1}^{k-1}N(u_i)\right)}

\title{Color Complexity of Recolorable Graph Exploration:\\
Upper and Lower Bounds via Block Structure}
\author{%
  Shoma Hiraoka\thanks{The University of Osaka} \and
  Shunsuke Imori\footnotemark[1] \and
  Shota Takahashi\thanks{Hosei University} \and
  Yuichi Sudo\footnotemark[2]
}
\date{}

\begin{document}
\maketitle

\begin{abstract}
    We study exploration of anonymous, port-free graphs by a single agent with no internal memory.
    To compensate for the lack of memory, the agent uses writable vertex colors as external memory.
    At each step, the agent observes its current color and the multiset of neighbor colors, recolors the current vertex, and requests a destination color.
    If several neighbors have that color, an adversary chooses the destination.
    From every starting vertex, the agent must visit all vertices, return to its start, and terminate there.
    Throughout, recoloring is unrestricted, and the color count includes the common initial color.
    Under this convention, previous work gave a six-color algorithm for general graphs and five-color algorithms for triangle-free graphs and for $\varphi$-free graphs, a class that contains all cacti.
    However, to our knowledge, no nontrivial color lower bound was known for unrestricted recoloring.

    We determine the optimal number of colors on two classes defined by block structure and prove the first nontrivial color lower bounds for unrestricted recoloring.
    First, a single three-color algorithm explores every tree and every simple cycle in $O(n)$ moves, and no algorithm with at most two colors explores $P_3$, the path on three vertices.
    Second, we give a four-color algorithm that explores every graph whose blocks are cycles or complete bipartite graphs in $O(n)$ moves, and we prove that no algorithm with at most three colors explores all subcubic pseudotrees.
    Hence four colors are optimal for every class between subcubic pseudotrees and this block-defined class.
    On cacti, this improves the previous five-color upper bound to a tight four.
    The lower bound reduces the possible initial actions by hand and rules out the remaining cases by a machine-checked SAT certificate on nine graphs with at most five vertices.
    Independent checkers verify the graph encodings, the hand reduction, and every imported trace and domain clause, and a DRAT checker verifies the refutation.
    Finally, we extend the known five-color algorithm for triangle-free graphs to graphs whose blocks are cliques or triangle-free, using $O(n\Delta)$ moves, where $\Delta$ is the maximum degree.
    The extension uses two conditions that let the semi-DFS cleanup restore temporary marks and return to the correct path endpoint against every adversarial choice.
    We prove that this class is exactly the class of graphs on which both conditions hold in every reachable state.
\end{abstract}

\noindent\textbf{2012 ACM Subject Classification:}
Theory of computation $\to$ Distributed algorithms;
Mathematics of computing $\to$ Graph theory

\medskip
\noindent\textbf{Keywords:}
graph exploration, oblivious agent, vertex recoloring, semi-DFS, cactus graphs,
computer-assisted proof

\bigskip

\section{Introduction}
\label{sec:introduction}

Graph exploration asks a single mobile agent to visit every vertex of an initially unknown connected graph and has been studied under many models and objectives~\cite{AKL+79,BGKP22,CFIKP08,DHK19,IKIM23,Koucky02,MPU17,PP99,PDD+96,Reingold08,SSKM22,SBN+15,SOK26,YWB03}.
Formulations differ in whether the agent must only visit every vertex, also return to its start, additionally detect completion and stop, or instead visit every vertex infinitely often.
Related tasks on anonymous networks, such as dispersion, gathering, uniform deployment, and black-hole search, pursue different objectives~\cite{AM18,DFPS06,DFPS07,DSS+08,KA19,KMS19,OKK+13,SMO+16,SNO+19,SSNK20,SSN+24}.
In this paper, we study exploration with return and termination.
From every starting vertex, the agent must visit all vertices, return to its start, and terminate there.

Whether this objective is achievable depends on the resources available to the agent.
Most exploration models rely on port numbers and internal memory, sometimes augmented with vertex identifiers, pebbles, or preprocessed labels~\cite{CFIKP08,DHK19,IKIM23,PP99,Reingold08}.
B\"ockenhauer, Frei, Unger, and Wehner~\cite{bockenhauer2023zero} ask how far exploration is possible without these standard resources, replacing them with writable vertex colors.
At each step, the agent observes the color of its current vertex and those of its neighbors, may recolor the current vertex, and either selects a destination color or terminates.
Unless the agent terminates, the adversary chooses a neighbor of the selected color and moves the agent there.
We call this setting the \emph{BFUW model}.
The central question in the BFUW model is how many colors suffice for exploration.

Throughout, recoloring is unrestricted, and every color count includes the common initial color.\footnote{B\"ockenhauer et al.~\cite{bockenhauer2023zero} do not count the initial color, but Takahashi, Kanaya, Hiraoka, Eguchi, and Sudo~\cite{takahashi2026recolorable} do.
    We follow Takahashi et al. and add one color when stating the bounds of B\"ockenhauer et al.}
When each vertex may be recolored at most once, B\"ockenhauer et al. prove that exactly $n$ colors are required on general $n$-vertex graphs and exactly four on trees.
These bounds concern the once-recolorable variant.
With unrestricted recoloring, they give an eight-color algorithm for general graphs that processes breadth-first layers in order, storing the layer index modulo a constant in the vertex colors.
Takahashi et al.~\cite{takahashi2026recolorable} instead introduce semi-DFS, a DFS variant that adds an unvisited neighbor of the current path endpoint only if that vertex is adjacent to no other vertex on the path.
Their implementation uses six colors on general graphs and five colors on triangle-free graphs.
Their separate five-color algorithm for $\varphi$-free graphs colors vertices according to distance modulo a constant, as does the eight-color algorithm.
A graph is $\varphi$-free if, for every vertex $v$ of degree at least four, each component of $G-v$ contains at most two neighbors of $v$.
This class contains every subcubic graph (i.e., a graph with maximum degree at most three) and every cactus.

\subsection{Our contribution}
We give upper and lower bounds on the number of colors required to explore several graph classes in the BFUW model with unrestricted recoloring.
Each upper bound below is achieved by a single algorithm that explores every graph in the corresponding class.
The color count includes the common initial color.

\emph{Three colors.}
We give a single three-color algorithm that explores every tree and every simple cycle in $O(n)$ moves.
We also prove that no algorithm with at most two colors explores $P_3$, the path on three vertices.
Thus the optimal number of colors for the class consisting of all trees and all simple cycles is three.

\emph{Four colors.}
A block is a maximal biconnected subgraph or a bridge.
We write $\gcb$ for the class of graphs whose blocks are cycles or complete bipartite graphs.
We give a four-color algorithm that explores every graph in this class in $O(n)$ moves.

We also give a computer-assisted proof that at least four colors are required for a single algorithm to explore every subcubic pseudotree.
Here, a graph is \emph{subcubic} if its maximum degree is at most three.
A \emph{pseudotree} is a connected graph with at most one cycle.
This lower bound implies that the optimal number of colors is four for every graph class that contains all subcubic pseudotrees and is contained in $\gcb$ because $\gcb$ contains all subcubic pseudotrees.
In particular, we improve the previous five-color upper bound for cacti to four and prove its optimality.

\emph{Five colors.}
We write $\gktf$ for the class of graphs whose blocks are cliques or triangle-free.
We extend the five-color algorithm of Takahashi et al.~\cite{takahashi2026recolorable} for triangle-free graphs to explore every graph in $\gktf$ with five colors in $O(n\Delta)$ moves.
Here, $\Delta$ denotes the maximum degree.

The class $\gktf$ and the class of $\varphi$-free graphs covered by a separate known five-color algorithm are incomparable.
A graph is \emph{unichord-free} if no cycle has exactly one chord~\cite{trotignon2010structure}.
Every unichord-free graph belongs to $\gktf$, so we also obtain a five-color upper bound for this class.

Our lower bounds are the first nontrivial color lower bounds for unrestricted recoloring.
Together with the known six-color algorithm for general graphs, they place the optimal number of colors for general graphs between four and six.
It remains open to prove our four-color lower bound without computer assistance.
Figure~\ref{fig:class-inclusions} summarizes the class inclusions, and Table~\ref{tab:oblivious-results} summarizes the upper and lower color bounds.

\begin{figure}[H]
    \centering
    \begin{tikzpicture}[
        x=1cm,y=1cm,font=\scriptsize,
        cls/.style={rectangle,draw,align=center,inner xsep=4pt,inner ysep=3pt},
        inc/.style={-{Latex[length=1.6mm]},semithick,shorten <=1pt,shorten >=1pt}
        ]
        \node[cls] (sub) at (0,0) {subcubic\\pseudotrees};
        \node[cls] (pseudo) at (2.7,0) {pseudotrees};
        \node[cls] (cactus) at (4.9,0) {cacti};
        \node[cls] (cb) at (7.1,0) {$\gcb$};
        \node[cls] (ktf) at (9.2,0) {$\gktf$};
        \node[cls] (deg2) at (-2.25,0) {$\Delta\leq2$};
        \node[cls] (trees) at (2.7,0.78) {trees};
        \node[cls] (deg3) at (2.7,-0.78) {$\Delta\leq3$};
        \node[cls] (phi) at (4.9,-0.78) {$\varphi$-free};
        \node[cls] (tfree) at (9.2,-0.78) {triangle-free};
        \draw[inc] (deg2.east)--(sub.west);
        \draw[inc] (trees.south)--(pseudo.north);
        \draw[inc] (sub.south)--(deg3.north west);
        \draw[inc] (deg3.east)--(phi.west);
        \draw[inc] (sub.east)--(pseudo.west);
        \draw[inc] (pseudo.east)--(cactus.west);
        \draw[inc] (cactus.east)--(cb.west);
        \draw[inc] (cb.east)--(ktf.west);
        \draw[inc] (cactus.south)--(phi.north);
        \draw[inc] (tfree.north)--(ktf.south);
        \draw[densely dashed,shorten <=1pt,shorten >=1pt]
        (phi.north east)--node[below,sloped]{incomparable} (ktf.south west);
    \end{tikzpicture}
    \caption{Relations among the graph classes used in this paper. Solid arrows denote proper inclusions, the dashed line denotes incomparability, and $\Delta$ denotes maximum degree.}
    \label{fig:class-inclusions}
\end{figure}

\begin{table}[htbp]
    \caption{Color bounds including the initial color, for a single algorithm per class. Lower bounds of four follow from Theorem~\ref{thm:subcubic-pseudotree-lower} by class containment.}
    \label{tab:oblivious-results}
    \centering
    \begin{tabular}{@{}lcc@{}}
        \toprule
        Graph class      & Lower bound                                       & Upper bound \\
        \midrule
        Trees and simple cycles
                         & $3$~\resultlemma{lem:p3-two-color}
                         & $3$~\resulttheorem{thm:tc3}                                     \\
        Cacti            & $4$~\resulttheorem{thm:subcubic-pseudotree-lower}
                         & $4$~\resulttheorem{thm:cb4}                                     \\
        $\gcb$           & $4$~\resulttheorem{thm:subcubic-pseudotree-lower}
                         & $4$~\resulttheorem{thm:cb4}                                     \\
        $\gktf$          & $4$~\resulttheorem{thm:subcubic-pseudotree-lower}
                         & $5$~\resulttheorem{thm:main}                                    \\
        $\varphi$-free   & $4$~\resulttheorem{thm:subcubic-pseudotree-lower}
                         & $5$~\cite{takahashi2026recolorable}                             \\
        Arbitrary graphs & $4$~\resulttheorem{thm:subcubic-pseudotree-lower}
                         & $6$~\cite{takahashi2026recolorable}                             \\
        \bottomrule
    \end{tabular}
\end{table}
\FloatBarrier

\subsection{Further Related Work}
Other exploration models provide different resources for navigation.

Aleliunas, Karp, Lipton, Lov\'asz, and Rackoff~\cite{AKL+79} show that a random walk covers every finite connected graph with probability one and introduce universal traversal sequences for port-numbered graphs.
Kouck\'y~\cite{Koucky02} introduced exploration sequences with backtracking, and Reingold~\cite{Reingold08} constructed universal exploration sequences in logarithmic space.
Neither approach gives deterministic exploration with return and termination without port numbers.

Deterministic exploration becomes possible with other forms of local information.
Cohen, Fraigniaud, Ilcinkas, Korman, and Peleg~\cite{CFIKP08} show that three-valued preprocessing labels let a finite automaton explore every graph.
Disser, Hackfeld, and Klimm~\cite{DHK19} prove that $\Theta(\log\log n)$ distinguishable pebbles, or the same number of constant-memory agents, are both necessary and sufficient to explore arbitrary port-numbered graphs.
Bojko, Gotfryd, Kowalski, and Paj\k{a}k~\cite{BGKP22} study trees with hidden incoming edges and bound the trade-offs among agent memory, per-vertex bits, and time.

Several works instead place persistent state directly at vertices.
Sudo, Baba, Nakamura, Ooshita, Kakugawa, and Masuzawa~\cite{SBN+15} study exploration with vertex whiteboards.
Priezzhev, Dhar, Dhar, and Krishnamurthy~\cite{PDD+96} introduced Eulerian walkers, which maintain a cyclic pointer at each vertex, Yanovski, Wagner, and Bruckstein~\cite{YWB03} apply them to perpetual patrolling, and Menc, Paj\k{a}k, and Uzna\'nski~\cite{MPU17} bound the time and space of rotor-router exploration.
Inoue, Kitamura, Izumi, and Masuzawa~\cite{IKIM23} instead study an agent with no persistent internal state (an oblivious agent) that can update persistent states at vertices.
Sudo, Ooshita, and Kamei~\cite{SOK26} study self-stabilizing exploration from arbitrary agent and vertex states, where the objective is to visit every vertex infinitely often.
For cacti with distinguishable cycles, Shimoyama, Sudo, Kakugawa, and Masuzawa~\cite{SSKM22} give an algorithm that visits every vertex infinitely often.

\section{Model and Preliminaries}
\label{sec:preliminaries}

Throughout, $G=(V,E)$ is a finite simple connected undirected graph with $n=|V|$ and maximum degree $\Delta$.
We write $N(v)=\{u\in V:\{u,v\}\in E\}$ for the neighborhood of a vertex $v$.
A single agent moves on $G$, whose vertices have colors from a finite set $C$.
Initially, the agent is at an arbitrary start $s$, and every vertex has the initial color $c_0\in C$.
The agent retains no internal state between actions and is therefore called \emph{oblivious}.

At each step, the agent observes only the color $c$ of its current vertex and the multiset $M$ of neighbor colors, including their multiplicities.
We write this observation as $(c;M)$.
The agent chooses a new current color and an action from this observation.
The action is a move to a neighbor of a specified color, $\cstay$ to remain at the current vertex, or $\cstop$ to terminate.
A deterministic \emph{rule function} $\xi:C\times\calM(C)\to C\times(C\cup\{\cstay,\cstop\})$ specifies these choices, where $\calM(C)$ denotes the finite multisets over $C$.
If $\xi(c,M)=(c',d)$, the agent recolors its current vertex with $c'$ and then performs action $d$.
This is one rule application.

If $d\in C$, an adversary chooses a neighbor of color $d$ as the destination.
At each step, the adversary sees all vertex colors and the agent position.
If no neighbor of color $d$ exists, this \emph{absent-color request} ends the process.
Since the graph has no self-loops, recoloring the current vertex does not change the neighborhood-color multiset $M$.
Recoloring is unrestricted, so a vertex may be recolored any number of times and may return from a noninitial color to $c_0$.

An exploration algorithm is specified by $\calA=(C,c_0,\xi)$.
The rule function is not given $n$, $\Delta$, a round counter, a distinguished start marker, vertex identifiers, the previous vertex, or incoming ports.

A \emph{configuration} $q=(v,\psi)$ records the agent position $v$ and vertex coloring $\psi:V\to C$ immediately before a rule application.
The observation in this configuration is $(\psi(v);M(v))$, where $M(v)=\{\!\{\psi(u):u\in N(v)\}\!\}$.
After a move or $\cstay$, the next configuration records the resulting coloring and position.
An \emph{execution} is a finite or infinite sequence of configurations obtained in this way from the initial configuration at $s$.
No configuration is recorded after a stop or an absent-color request.
An execution is \emph{maximal} if it is infinite or ends by issuing $\cstop$ or requesting an absent color.

Under the exploration objective of B\"ockenhauer et al.~\cite{bockenhauer2023zero}, a maximal execution from $s$ \emph{succeeds} if it visits every vertex and issues $\cstop$ at $s$.
A maximal execution \emph{fails} if it runs forever, requests an absent color, or issues $\cstop$ before visiting every vertex or at a vertex other than $s$.
An algorithm explores $G$ if, for every start $s$ and every sequence of adversarial choices, the resulting maximal execution succeeds.
For a fixed algorithm and graph, a configuration or observation is \emph{reachable} if it occurs in an execution from some starting vertex.

The number of colors is $|C|$, including the initial color $c_0$.
The \emph{color complexity} of a graph class is the minimum $|C|$ for which a single algorithm with color set $C$ explores every graph in the class.
An algorithm with fewer than $k$ colors can be extended to a color set of size $k$ by leaving the additional colors unused.
Running-time bounds count rule applications, including $\cstay$ actions.

For a path $P=(u_1,\dots,u_k)$, we write $V(P)=\{u_1,\dots,u_k\}$ for its vertex set.
A path is induced if no two nonconsecutive vertices on the path are adjacent.
A vertex $v$ is a \emph{cut vertex} of a graph $H$ if $H-v$ has more connected components than $H$.
A \emph{block} of $G$ is a maximal connected subgraph with at least two vertices and no cut vertex.
Equivalently, a block is a maximal biconnected subgraph or a bridge on its two endpoints.
Every edge belongs to exactly one block, every cycle lies in one block, and the cut vertices of $G$ are exactly the vertices in at least two blocks.

For pairwise distinct colors $c_1,\dots,c_m$, we write $c_1,\dots,c_m\in M$ if each $c_i$ occurs in $M$ and $c_1,\dots,c_m\notin M$ if none occurs.

\begin{jlemma}
    \label{lem:stay-elimination}
    A graph family is explorable with a given color set if and only if it is explorable by a rule that never uses $\cstay$ at a reachable observation.
\end{jlemma}

\begin{proof}
    During consecutive $\cstay$ actions, the agent position and neighborhood multiset $M$ remain fixed.
    The starting observation $(c;M)$ determines the sequence, which must end with a move or stop because the rule succeeds.
    Replace each reachable such sequence by its final recoloring and move or stop, retaining every other output.
    The resulting executions are exactly the original executions with their $\cstay$ steps deleted.
    The converse holds because a rule that never stays is allowed by the original model.
\end{proof}

\begin{jlemma}
    \label{lem:execution-merging}
    If two executions of the same rule on the same graph, started at distinct vertices, reach the same configuration before their next rule application, then some starting vertex and sequence of adversarial choices yield a failing maximal execution of the rule.
\end{jlemma}

\begin{proof}
    After the executions reach the same configuration, the adversary can choose the same destinations, so determinism keeps all later configurations identical.
    An infinite continuation or an absent-color request makes both executions fail.
    Otherwise they stop at the same vertex, which cannot be both distinct starts.
\end{proof}

\section{Three Colors}
\label{sec:tight-three-color}

In this section, we present the three-color algorithm $\mathcal A_{\mathrm{TC3}}$ (Algorithm~\ref{alg:tree-cycle-three}) for trees and simple cycles and show that two colors do not suffice for $P_3$, the path on three vertices.

\subsection{Trees and simple cycles}

\begin{algorithm}[!t]
    \caption{Three-Color Exploration of Trees and Simple Cycles $\mathcal A_{\mathrm{TC3}}$}
    \label{alg:tree-cycle-three}
    \Colors{
        $C_{\mathrm{TC3}}=\{\cinit,\clzero,\clone\}$ and $c_0=\cinit$
    }
    \RuleFunction{\\
        \begingroup
        \setlength{\arraycolsep}{3pt}
        \begin{align*}
            \xi_{\mathrm{TC3}}(\cinit,M)  & =
            \begin{cases}
                (\clone,\cstop)
                 & \mif \cinit,\clzero,\clone\notin M
                \hfill\ (\rlabel{D1})                            \\
                (\clzero,\cinit)
                 & \melif \cinit\in M\land\clzero,\clone\notin M
                \hfill\ (\rlabel{D2})                            \\
                (\clone,\cinit)
                 & \melif \cinit,\clzero\in M
                \hfill\ (\rlabel{D3})                            \\
                (\cinit,\clone)
                 & \melif \cinit,\clone\in M
                \hfill\ (\rlabel{D4})                            \\
                (\clone,\clone)
                 & \melif \clone\in M
                \hfill\ (\rlabel{D5})                            \\
                (\clone,\clzero)
                 & \melif \clzero\in M
                \hfill\ (\rlabel{D6})
            \end{cases}                                       \\
            \xi_{\mathrm{TC3}}(\clzero,M) & =
            \begin{cases}
                (\clzero,\cinit) & \mif \cinit,\clone\in M \hfill\ (\rlabel{D7})                       \\
                (\clone,\clzero) & \melif \clzero,\clone\in M \hfill\ (\rlabel{D8})                    \\
                (\clone,\cstop)  & \melif \clone\in M\land\cinit,\clzero\notin M \hfill\ (\rlabel{D9})
            \end{cases} \\
            \xi_{\mathrm{TC3}}(\clone,M)  & =
            \begin{cases}
                (\clzero,\cinit) & \mif \cinit,\clzero\in M \hfill\ (\rlabel{D10})   \\
                (\clone,\clzero) & \melif \clzero,\clone\in M \hfill\ (\rlabel{D11})
            \end{cases}
        \end{align*}
        \endgroup}
\end{algorithm}

On a tree rooted at the start, an entry into an all-$\cinit$ child subtree from a $\clzero$ parent returns with the subtree colored $\clone$.
On a simple cycle, the first move fixes a direction.
The agent colors the forward route $\clzero$, reverses direction, and colors the return route $\clone$ before stopping at the start.
Outputs omitted from a rule table may be assigned arbitrarily; each correctness proof shows that the corresponding observations are unreachable.

\begin{jtheorem}
    \label{thm:tc3}
    Algorithm~\ref{alg:tree-cycle-three} explores every tree and every simple cycle using three colors and $O(n)$ moves.
\end{jtheorem}

\begin{proof}
    \emph{Case 1 (tree).}
    We root the tree at the start $s$.
    If $s$ is isolated, \rlabel{D1} stops immediately.
    Otherwise \rlabel{D2} recolors $s$ to $\clzero$ and moves to an adversarially chosen neighbor.

    For each nonroot vertex $v$ with parent $p$, let $T_v$ denote the subtree rooted at $v$.
    We prove the following claim by induction on its height.
    The induction claim assumes that the agent is at $v$, the parent $p$ is in $\clzero$, and all of $T_v$ is in $\cinit$.
    From this configuration, the agent visits every vertex of $T_v$, colors them all with $\clone$, and returns to $p$, whose color remains $\clzero$.
    Except for the final move to $p$, the agent stays inside $T_v$.
    If $v$ is a leaf, \rlabel{D6} does exactly this.

    We next assume that $v$ is not a leaf.
    \rlabel{D3} recolors $v$ to $\clone$ and enters an adversarially chosen $\cinit$-colored child $x$.
    If $x$ is a leaf, the agent sees only the $\clone$-colored vertex $v$ there, so \rlabel{D5} recolors $x$ with $\clone$ and returns to $v$.
    Otherwise the agent sees both $\cinit$ and $\clone$, so \rlabel{D4} returns immediately to $v$ and leaves every vertex of $T_x$ colored $\cinit$.
    Hence the agent is back at $v$, and every child subtree is either entirely $\cinit$ or entirely $\clone$.

    If a child subtree still has all vertices colored $\cinit$, \rlabel{D10} first changes $v$ to $\clzero$ and enters one such subtree.
    On later iterations, \rlabel{D7} keeps $v$ colored $\clzero$ and enters the next such subtree.
    By the induction hypothesis, the agent visits every vertex of that subtree, colors them all with $\clone$, and returns to $v$.
    This decreases the number of child subtrees whose vertices all have color $\cinit$ by one.
    When none remains, \rlabel{D11} if $v$ is $\clone$, or \rlabel{D8} if $v$ is $\clzero$, recolors $v$ with $\clone$ and returns to $p$.
    Until the final return every move stays at $v$ or inside a child subtree, so $p$ remains $\clzero$.

    After \rlabel{D2}, the claim applies because $s$ is $\clzero$ and the entered child subtree is entirely $\cinit$.
    After each return, \rlabel{D7} enters another $\cinit$-colored child if one remains.
    Otherwise \rlabel{D9} recolors $s$ with $\clone$ and stops there.
    Entering a child subtree and returning to its parent traverses the edge between its root and parent once in each direction.
    The test by \rlabels{D3--D4} traverses the same edge at most once more in each direction.
    Hence the execution uses $O(n)$ moves.

    \emph{Case 2 (simple cycle).}
    Let $v_1$ be the first selected destination, and list the cycle vertices in order as $v_0=s,v_1,\dots,v_{n-1},v_0$.
    Immediately after \rlabel{D2} and after each complete sequence \rlabels{D3--D4--D10}, the agent is at $v_i$, the vertices $v_0,\dots,v_{i-1}$ are $\clzero$, and $v_i$ and every vertex in the rest of the cycle have color $\cinit$.
    For $1\le i\le n-3$, the sequence \rlabels{D3--D4--D10} moves the agent to $v_{i+1}$ and restores the same condition with $i+1$ in place of $i$.
    For $i=n-2$, the sequence \rlabels{D3--D5--D11} visits the last vertex $v_{n-1}$ and returns to $v_{n-3}$, reversing direction.
    Each application of \rlabel{D8} changes one $\clzero$ vertex to $\clone$ and moves one step closer to $s$.
    \rlabel{D9} stops at $s$.
    After \rlabel{D2}, every requested destination is unique.
    There are $O(1)$ rule applications per vertex, giving $O(n)$ moves in total.

    A displayed rule applies at every stage described above, so every omitted observation is unreachable.
\end{proof}

\subsection{The two-color lower bound and tightness}

\begin{jlemma}
    \label{lem:p3-two-color}
    The path $P_3$ cannot be explored with at most two colors.
\end{jlemma}

\begin{proof}
    We write $P_3=(v_1,v_2,v_3)$.
    By adding unused colors if necessary and renaming colors, we may assume w.l.o.g.~that $C=\{0,1\}$, with $0$ as the initial color.
    The notation $a\to b$ means that the agent recolors the current vertex to $a$ and requests color $b$.
    The notation $(v_i;a,b,c)$ denotes the configuration with the agent at $v_i$ and colors $(a,b,c)$ on $(v_1,v_2,v_3)$.
    By Lemma~\ref{lem:stay-elimination}, we may assume that no reachable observation uses $\cstay$.
    Initially, stopping leaves vertices unvisited, and only color $0$ is available as a destination.
    Moving without recoloring reaches the initial configuration for a start at the destination.
    Hence Lemma~\ref{lem:execution-merging} forces $(0;\{0\})\mapsto 1\to0$ and $(0;\{0,0\})\mapsto 1\to0$.
    We set $\alpha=\xi(0;\{1\})$ and $\beta=\xi(0;\{0,1\})$.

    We start at $v_2$, and the adversary chooses $v_1$ after the forced initial action.
    The agent then observes $(0;\{1\})$ at $v_1$.
    The vertex $v_3$ remains unvisited.
    Hence $\alpha$ cannot stop and is either $0\to1$ or $1\to1$.
    We first suppose $\alpha=0\to1$.
    The execution reaches $(v_2;0,1,0)$, with observation $(1;\{0,0\})$.
    There, $0\to0$ reaches an endpoint with every vertex colored 0, which is also the initial configuration for a start at that endpoint.
    With output $1\to0$, the adversary can choose $v_1$ as the destination, after which $\alpha$ returns the agent to $(v_2;0,1,0)$.
    Repeating this choice prevents termination.
    Every other action either stops with $v_3$ unvisited or requests an absent color.
    Thus $\alpha=1\to1$.

    Starting at $v_1$, the forced first action reaches observation $(0;\{0,1\})$ at $v_2$.
    The vertex $v_3$ remains unvisited.
    Hence $\beta$ cannot stop, and it remains to consider its four possible move values.

    If $\beta=0\to0$, the executions started at $v_1$ and $v_3$ both reach $(v_2;1,0,1)$.
    If $\beta=1\to0$, the same two executions, using $\alpha=1\to1$, both reach $(v_2;1,1,1)$.

    If $\beta=0\to1$, the execution starting at $v_1$ reaches $(v_1;1,0,0)$ with observation $(1;\{0\})$.
    At this observation, output $0\to0$ reaches the initial configuration for a start at $v_2$.
    Output $1\to0$ followed by $\beta$ returns to $(v_1;1,0,0)$, so these two actions repeat without termination.

    Finally, if $\beta=1\to1$, the execution starting at $v_1$ reaches $(v_1;1,1,0)$ with observation $(1;\{1\})$.
    With output $0\to1$, the executions starting at $v_1$ and $v_3$ both reach $(v_2;0,1,0)$.
    With output $1\to1$, the execution starting at $v_1$ and the execution starting at $v_2$ after its forced first move and $\alpha$ both reach $(v_2;1,1,0)$.
    At observations $(1;\{0\})$ and $(1;\{1\})$, stopping leaves $v_3$ unvisited.
    Requesting an absent color also fails.
    Thus every two-color rule either has a failing execution or allows executions from distinct starts to reach the same configuration.
    In the latter case, Lemma~\ref{lem:execution-merging} also gives a failing maximal execution, so no rule with at most two colors explores $P_3$.
\end{proof}

Since $P_3$ is a tree, Theorem~\ref{thm:tc3} and Lemma~\ref{lem:p3-two-color} show that the class of all trees and simple cycles has color complexity exactly three.

\section{Four-Color Exploration}
\label{sec:uniform-return}

In this section, we present the four-color algorithm $\mathcal A_{\mathrm{C4}}$ (Algorithm~\ref{alg:cactus4}).
It explores every graph in $\gcb$ in $O(n)$ moves.
Each edge of $K_{1,q}$ forms a bridge block $K_{1,1}$, so it suffices to consider complete bipartite blocks $K_{1,1}$ and $K_{p,q}$ with $p,q\ge2$.

\begingroup
\begin{algorithm}[!t]
    \caption{Four-Color Exploration $\mathcal A_{\mathrm{C4}}$}
    \label{alg:cactus4}
    \Colors{
        $C_{\mathrm{C4}}=\{\cinit,\cfin,\cfront,\cpath\}$ and
        $c_0=\cinit$
    }
    \RuleFunction{\\
        \begin{align*}
            \xi_{\mathrm{C4}}(c,M) & =
            \begin{cases}
                (\cfin,\cstop)   & \mif \cinit,\cfront,\cpath\notin M
                \quad\hfill\ (\rlabel{C1})                             \\
                (\cfront,\cinit) & \melif \cfin,\cfront,\cpath\notin M
                \quad\hfill\ (\rlabel{C2})                             \\
                \eta_c(M)        & \moth
            \end{cases}
            \quad(c\neq\cfin),                              \\
            \eta_{\cinit}(M)       & =
            \begin{cases}
                (\cfin,\cpath)    & \mif \cfront,\cinit\notin M
                \hfill\ (\rlabel{C3})                           \\
                (\cfin,\cfront)   & \melif \cinit\notin M
                \hfill\ (\rlabel{C4})                           \\
                (\cpath,\cinit)   & \melif \cfront,\cpath\in M
                \hfill\ (\rlabel{C5})                           \\
                (\cfront,\cinit)  & \melif \cpath\in M
                \hfill\ (\rlabel{C6})                           \\
                (\cfront,\cfront) & \melif \cfront\in M
                \hfill\ (\rlabel{C7})
            \end{cases} \\
            \eta_{\cfront}(M)      & =
            \begin{cases}
                (\cpath,\cfront) & \mif \cfront\in M
                \hfill\ (\rlabel{C8})                 \\
                (\cfront,\cinit) & \melif \cinit\in M
                \hfill\ (\rlabel{C9})                 \\
                (\cfin,\cpath)   & \moth
                \hfill\ (\rlabel{C10})
            \end{cases}           \\
            \eta_{\cpath}(M)       & =
            \begin{cases}
                (\cfront,\cinit) & \mif \cfront\notin M\land\cinit\in M
                \quad\hfill\ (\rlabel{C11})                             \\
                (\cpath,\cinit)  & \melif \cfront,\cinit\in M
                \hfill\ (\rlabel{C12})                                  \\
                (\cfin,\cfront)  & \melif \cfront\in M
                \hfill\ (\rlabel{C13})                                  \\
                (\cfin,\cpath)   & \moth
                \hfill\ (\rlabel{C14})
            \end{cases}
        \end{align*}
    }
\end{algorithm}
\endgroup

Moves to initial-color vertices leave their sources in $\cfront$ or $\cpath$, distinguishing their behavior on return.
Recoloring a vertex $\cfin$ \emph{completes} it.
Since $\cfin$ is never requested, completed vertices are never revisited and need no rule-table entry.
At any noninitial arrival at an $\cinit$ vertex, its source remains a $\cfront$ or $\cpath$ neighbor.
Thus an $\cinit$-vertex observation whose neighbor colors are exactly $\cinit$ and $\cfin$ is unreachable and may have an arbitrary output.
Figure~\ref{fig:c4-dynamics} shows two executions.

To analyze the exploration, we use the following tree to represent the connections between blocks.
For nonisolated $s$, the \emph{block-cut tree} $\mathcal T_G$ has one node for each block and one node for each cut vertex, with adjacency given by containment.
We root $\mathcal T_G$ at the cut-vertex node $s$ when it exists and at the unique block containing $s$ otherwise.
For a nonroot block $B$, let $\rho(B)$ be its parent cut vertex.
For a root block $B_0$, set $\rho(B_0)=s$.
Let $G_B$ be the union of the blocks in the subtree rooted at $B$.
We explore $G_B$ from $r=\rho(B)$, called the \emph{parent attachment}.
A \emph{child block} of $B$ is a block node $D$ whose parent cut-vertex node is a child of $B$ in $\mathcal T_G$.
We call $v=\rho(D)\in V(B)$ its \emph{child attachment}.
The color of a parent attachment at entry means its color after the recoloring in the move into the block.

\begin{samepage}
    \begin{jlemma}
        \label{lem:cb-block-call}
        Let $B$ be a block with parent attachment $r$.
        During an execution of Algorithm~\ref{alg:cactus4}, suppose that the agent moves from $r$ to an $\cinit$-colored vertex of $B$ according to its rules, and every vertex of $G_B-r$ still has color $\cinit$.
        The color of $r$ at entry is $\cfront$ or $\cpath$.
        For every adversarial choice, the agent remains in $G_B$ and never requests an absent color until its final return to $r$.
        In finite time, it colors every vertex of $G_B-r$ with $\cfin$ and makes its final return to $r$.
        \begin{enumerate}[(i)]
            \item If $r$ is $\cpath$ at entry, the agent returns with $r$ in $\cpath$.
            \item If $r$ is $\cfront$ at entry, the agent visits $r$ at most once before the final return.
                  Such a visit immediately follows entry.
                  The agent recolors $r$ to $\cpath$ and returns to the vertex just entered.
                  At the final return, $r$ has color $\cfront$ or $\cpath$.
        \end{enumerate}
    \end{jlemma}
\end{samepage}

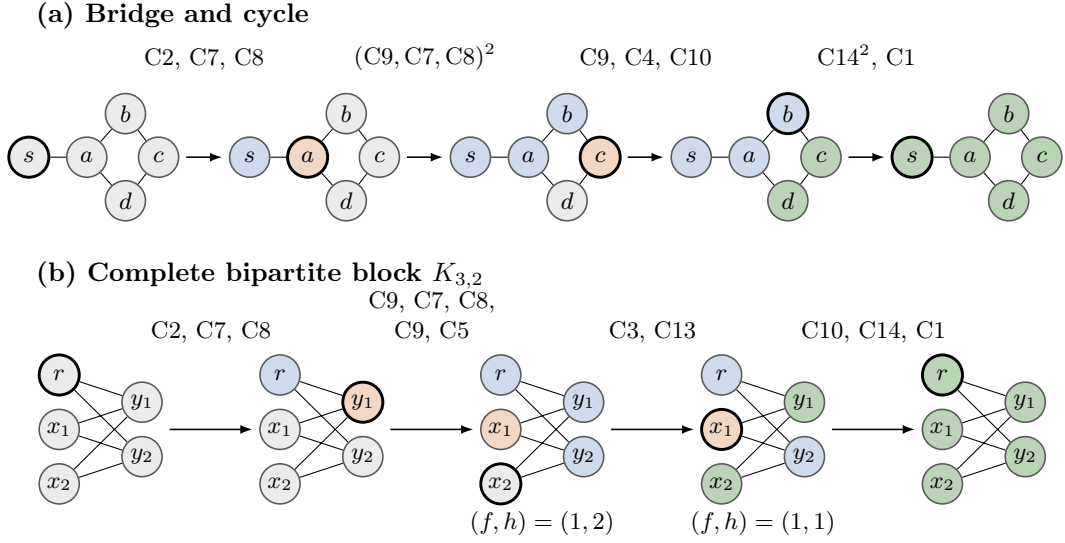
\begin{figure}[!t]
    \centering
    \definecolor{c4init}{RGB}{235,235,235}
    \definecolor{c4front}{RGB}{213,122,54}
    \definecolor{c4path}{RGB}{64,111,171}
    \definecolor{c4fin}{RGB}{104,159,96}
    \begin{tikzpicture}[
        x=0.73cm,y=0.80cm,font=\footnotesize,
        v/.style={circle,draw=black!65,semithick,minimum size=5.3mm,inner sep=0pt},
        iv/.style={v,fill=c4init}, Fv/.style={v,fill=c4front!30}, Pv/.style={v,fill=c4path!25}, fv/.style={v,fill=c4fin!45},
        here/.style={draw=black,line width=1.1pt},
        flow/.style={-{Latex[length=1.6mm]},semithick},
        translabel/.style={fill=white,inner sep=1.2pt,align=center,yshift=11mm}
        ]
        \node[font=\small\bfseries,anchor=west] at (-1.1,2.35) {(a) Bridge and cycle};
        \foreach \sh/\cs/\ca/\cb/\cc/\cd/\hs/\ha/\hb/\hc/\hd in {
                0/iv/iv/iv/iv/iv/here//////,
                4.0/Pv/Fv/iv/iv/iv//here////,
                8.0/Pv/Pv/Pv/Fv/iv////here//,
                12.0/Pv/Pv/Pv/fv/fv///here///,
                16.0/fv/fv/fv/fv/fv/here//////}{
                \begin{scope}[shift={(\sh,0)}]
                    \node[\cs,\hs] (s) at (-1.05,0) {$s$};
                    \node[\ca,\ha] (a) at (0,0) {$a$};
                    \node[\cb,\hb] (b) at (0.7,0.72) {$b$};
                    \node[\cc,\hc] (c) at (1.3,0) {$c$};
                    \node[\cd,\hd] (d) at (0.7,-0.72) {$d$};
                    \draw (s)--(a)--(b)--(c)--(d)--(a);
                \end{scope}}
        \draw[flow] (1.8,0)--node[above,translabel]{C2, C7, C8} (2.45,0);
        \draw[flow] (5.8,0)--node[above,translabel]{$(\mathrm{C9,C7,C8})^2$} (6.45,0);
        \draw[flow] (9.8,0)--node[above,translabel]{C9, C4, C10} (10.45,0);
        \draw[flow] (13.8,0)--node[above,translabel]{$\mathrm{C14}^2$, C1} (14.45,0);

        \node[font=\small\bfseries,anchor=west] at (-1.1,-1.95) {(b) Complete bipartite block $K_{3,2}$};
        \foreach \sh/\cr/\cxa/\cxb/\cya/\cyb/\hr/\hxa/\hxb/\hya/\hyb/\fh in {
        0/iv/iv/iv/iv/iv/here//////{},
        4.0/Pv/iv/iv/Fv/iv////here//{},
        8.0/Pv/Fv/iv/Pv/Pv///here///{$(f,h)=(1,2)$},
        12.0/Pv/Fv/fv/fv/Pv//here////{$(f,h)=(1,1)$},
        16.0/fv/fv/fv/fv/fv/here//////{}}{
        \begin{scope}[shift={(\sh,-4.5)}]
            \node[\cr,\hr] (r) at (-0.5,0.9) {$r$};
            \node[\cxa,\hxa] (xa) at (-0.5,0) {$x_1$};
            \node[\cxb,\hxb] (xb) at (-0.5,-0.9) {$x_2$};
            \node[\cya,\hya] (ya) at (1.0,0.45) {$y_1$};
            \node[\cyb,\hyb] (yb) at (1.0,-0.45) {$y_2$};
            \draw (r)--(ya) (r)--(yb) (xa)--(ya) (xa)--(yb) (xb)--(ya) (xb)--(yb);
            \node[anchor=north,inner sep=1pt] at (0.25,-1.3) {\fh};
        \end{scope}}
        \draw[flow] (1.5,-4.5)--node[above,translabel]{C2, C7, C8} (3.0,-4.5);
        \draw[flow] (5.5,-4.5)--node[above,translabel]{C9, C7, C8,\\C9, C5} (7.0,-4.5);
        \draw[flow] (9.5,-4.5)--node[above,translabel]{C3, C13} (11.0,-4.5);
        \draw[flow] (13.5,-4.5)--node[above,translabel]{C10, C14, C1} (15.0,-4.5);
    \end{tikzpicture}
    \caption{Traces of Algorithm~\ref{alg:cactus4} for one adversary.
        Gray, orange, blue, and green denote $\cinit$, $\cfront$, $\cpath$, and $\cfin$.
        Heavy outlines mark the agent.
        In (b), $r$ is the starting vertex, $X=\{r,x_1,x_2\}$ and $Y=\{y_1,y_2\}$, and $(f,h)$ counts the $\cfront$ vertices of $X-r$ and the $\cpath$ vertices of $Y$.}
    \label{fig:c4-dynamics}
\end{figure}

\begin{proof}
    We use induction on the height of $B$ in the block-cut tree.
    We first show how the induction hypothesis allows us to omit exploration below child blocks, then analyze the moves within $B$.
    A move to an $\cinit$ vertex leaves its source in $\cfront$ under \rlabels{C2, C6, C9, C11}, which require no $\cfront$ neighbor, or in $\cpath$ under \rlabels{C5, C12}, which require one.

    \emph{Child blocks.}
    On first entry into a child block $D$ from $v=\rho(D)$, all of $G_D-v$ is $\cinit$-colored by the block-cut-tree structure.
    By induction, the agent completes $G_D-v$ and returns to $v$ in finite time without changing $B-v$.
    The resulting $\cfin$-colored neighbors do not affect \rlabels{C9--C14}.
    After return from a $\cpath$ entry, \rlabel{C12} enters another $\cinit$ neighbor if one remains; otherwise \rlabel{C13} completes $v$ toward the $\cfront$ side of $B$.
    After return from a $\cfront$ entry, \rlabel{C9} or \rlabel{C11} enters another $\cinit$ neighbor if one remains; otherwise \rlabel{C10} or \rlabel{C14} completes $v$ toward the $\cpath$ side.
    Thus child exploration preserves the choice between advancing and returning in $B$.
    Each child exploration completes every vertex other than its attachment, so only finitely many occur.
    We omit their moves below.
    Between these explorations, every vertex of $G_B-B$ has color $\cinit$ or $\cfin$.

    \emph{Initial $\cfront$ entry.}
    For an initial $\cfront$ entry at an $\cinit$-colored vertex $v$, the entering rule is applicable only when $r$ has no $\cfront$ neighbor.
    If $v$ has an $\cinit$-colored neighbor, \rlabel{C7} makes $v$ the unique $\cfront$ neighbor of $r$ and moves back to $r$.
    \rlabel{C8} must therefore return to $v$.
    Together these rules change $r$ from $\cfront$ to $\cpath$ and leave $v$ in $\cfront$.
    If $v$ has no $\cinit$-colored neighbor, \rlabel{C4} completes $v$ and returns directly to $r$.

    \emph{Bridge.}
    For a bridge $\{r,v\}$, the explorations below child blocks at $v$ finish before the agent colors $v$ with $\cfin$ and returns to $r$.
    A $\cpath$ entry uses \rlabel{C3} directly, or \rlabel{C6} followed by \rlabel{C10} or \rlabel{C14} after explorations below child blocks.
    A $\cfront$ entry uses \rlabel{C4} directly, or \rlabels{C7--C8} first change $r$ from $\cfront$ to $\cpath$ before \rlabel{C10} or \rlabel{C14} eventually returns there.

    \emph{Simple cycle.}
    We write the cycle as $r,v_1,\dots,v_m,r$ in the direction fixed on entry.
    Between the rule sequences described below, the cycle vertices recolored so far form an arc from $r$ whose far endpoint is $\cfront$ and whose other vertices are $\cpath$.
    All vertices in the interior of the other arc still have color $\cinit$.
    A $\cpath$ entry establishes this form when \rlabel{C6} recolors the first vertex.
    For a $\cfront$ entry, \rlabels{C7--C8} first change $r$ to $\cpath$ and then reach the same form.
    Each advance uses \rlabel{C9} followed by \rlabels{C7--C8}, recoloring the old endpoint $\cpath$ and moving the $\cfront$ endpoint one step along the arc with $\cinit$-colored interior vertices.
    The \rlabel{C9} move may instead enter a child block at the $\cfront$ endpoint.
    By the induction hypothesis, the agent returns to the endpoint with its color $\cfront$ or $\cpath$, and \rlabel{C9} or \rlabel{C11} repeats the move because no neighbor of the endpoint is $\cfront$.
    Once the move reaches the next cycle vertex, \rlabels{C7--C8} restore the form above.
    When no $\cinit$-colored cycle vertex remains, \rlabel{C4} or \rlabel{C13}, followed by \rlabel{C10} or \rlabel{C14}, makes the agent reverse direction.
    \rlabel{C14} then completes one $\cpath$ vertex per step until it returns to $r$.
    \rlabel{C11} first processes any $\cinit$-colored child encountered during the return phase.
    The agent visits $r$ before the final return only during the initial \rlabels{C7--C8}.

    \emph{Complete bipartite block.}
    We denote the bipartition of $B=K_{p,q}$ by $(X,Y)$, where $r\in X$ and $p,q\ge2$.
    Between the rule sequences below, $r$ has color $\cpath$, every vertex of $X-r$ has color $\cinit$, $\cfront$, or $\cfin$, and every vertex of $Y$ has color $\cinit$, $\cpath$, or $\cfin$, except in the final step described last.
    The pair $(f,h)$ counts the $\cfront$ vertices in $X-r$ and the $\cpath$ vertices in $Y$.
    Same-colored vertices in one partite set have identical neighborhoods in $B$, so every adversarial choice gives the same count change.
    A $\cpath$ entry reaches a state with $r$ and one vertex of $Y$ in $\cpath$ and one vertex of $X-r$ in $\cfront$ by \rlabel{C6} and \rlabels{C7--C8}.
    A $\cfront$ entry reaches the same state after \rlabels{C7--C8} change $r$ to $\cpath$, followed by \rlabel{C9} and another application of \rlabels{C7--C8}.
    In both cases, $(f,h)=(1,1)$.

    \emph{Advance.}
    When both sides contain $\cinit$-colored vertices, each iteration uses \rlabel{C9}, \rlabel{C5}, and \rlabel{C6} to change one vertex of $X-r$ to $\cfront$ and one vertex of $Y$ to $\cpath$, increasing both counts by one.
    \rlabel{C11} and \rlabel{C12} resume the same process after explorations below child blocks.

    \emph{Exhaustion of one side.}
    If $Y$ runs out of $\cinit$-colored vertices first, the agent completes each remaining $\cinit$-colored vertex of $X$ and returns to $Y$, leaving $(f,h)=(q-1,q)$.
    This uses \rlabel{C3}, or \rlabel{C6} followed by \rlabel{C10} or \rlabel{C14} if explorations below child blocks intervene.
    \rlabel{C12} selects another $\cinit$-colored vertex of $X$ until none remains.
    If $X$ runs out first, the agent similarly completes each remaining $\cinit$-colored vertex of $Y$ and returns to $X-r$, leaving $(f,h)=(p-1,p-1)$.
    This uses \rlabel{C4}, or \rlabel{C5} followed by \rlabel{C13} after explorations below child blocks.
    \rlabel{C9} or \rlabel{C11} selects the next $\cinit$-colored vertex.

    \emph{Return.}
    When $B$ has no $\cinit$ vertex, the return phase starts with $(f,h)=(k,k)$ or $(k-1,k)$ for some $k\ge1$.
    Without child calls, \rlabel{C10} maps $(k,k)$ to $(k-1,k)$, and \rlabel{C13} maps $(k-1,k)$ to $(k-1,k-1)$ for $k>1$.
    These moves reach $(0,1)$ with $h\ge1$ throughout.
    Child calls during the return phase change the counts only temporarily.
    At a vertex of $Y$ with $f\ge1$, child calls use \rlabel{C12} and preserve the attachment's $\cpath$ color.
    A child call at $x\in X-r$ may return with $x$ in $\cpath$ instead of $\cfront$, decreasing $f$ by one.
    If an $\cinit$ neighbor remains, \rlabel{C11} restores $x$ to $\cfront$ and restores the counts.
    Otherwise \rlabel{C14} completes $x$ and moves to a $\cpath$ vertex of $Y$, which exists because $h\ge1$.
    At $y\in Y$ with $(f,h)=(0,1)$, remaining child calls use \rlabel{C11}, then \rlabel{C9} or \rlabel{C11}.
    These calls may change $y$ to $\cfront$ and the counts to $(0,0)$.
    This is the final step that leaves the color pattern above.
    After all children finish, \rlabel{C10} or \rlabel{C14} completes $y$ and returns to $r$, its unique $\cpath$ neighbor.
    Each return move completes one vertex of $B-r$, so the phase terminates at $r$.

    The three cases prove the lemma for $B$ whenever it holds for its children.
    A leaf block has no explorations below child blocks, so induction on block height completes the proof.
\end{proof}

\begin{jtheorem}
    \label{thm:cb4}
    Every graph in $\gcb$ can be explored by Algorithm~\ref{alg:cactus4} using four colors and $O(n)$ moves.
\end{jtheorem}

\begin{proof}
    Each first block entry follows an algorithm rule, and the initial-color premise of Lemma~\ref{lem:cb-block-call} holds because its parent attachment is the unique access to its block-cut subtree and only the current vertex is recolored.
    If $s$ is isolated, \rlabel{C1} completes it and stops.
    Otherwise each move from $s$ to an $\cinit$-colored neighbor enters an unprocessed block containing $s$.
    By Lemma~\ref{lem:cb-block-call}, the agent then completes all vertices other than $s$ in the block-cut subtree below that block and returns to $s$.
    That subtree is never entered again, so after finitely many such explorations every vertex other than $s$ is $\cfin$.
    \rlabel{C1} then completes $s$ and stops.

    We associate each move whose destination is an $\cinit$ vertex with that destination.
    We associate a pair \rlabels{C7--C8} with the $\cinit$-colored vertex at which \rlabel{C7} is executed, and a return move with the vertex that it changes to $\cfin$.
    The algorithm never recolors a vertex $\cinit$ after it leaves that color and never changes a $\cfin$ vertex.
    Thus only $O(1)$ moves are associated with each vertex, and the agent makes $O(n)$ moves.
\end{proof}

\section{Five-Color Exploration}
\label{sec:ktf-five}

In this section, we present the five-color algorithm $\aus$ (Algorithm~\ref{alg:uniform-five}).
It explores every graph in $\gktf$ in $O(n\Delta)$ moves.
In $\gktf$, every block containing a triangle is a clique.

\begingroup
\begin{algorithm}[!t]
    \caption{Five-Color Exploration $\aus$}
    \label{alg:uniform-five}
    \Colors{
        $\cus=\{\cinit,\cpath,\cneigh,\cfin,\cheadd\}$ and
        $c_0=\cinit$
    }
    \RuleFunction{\\
        \begin{align*}
            \xius(\cinit,M)  & \!\!=
            \begin{cases}
                (\cpath,\cinit)   & \mif \cinit\in M\land \cpath,\cneigh,\cfin,\cheadd\notin M
                \hfill\ (\rlabel{U1})                                                          \\
                (\cpath,\cstay)   & \melif \cneigh\in M\land\cheadd\notin M
                \hfill\ (\rlabel{U2})                                                          \\
                (\cneigh,\cheadd) & \melif \cheadd\in M\land\cpath\in M
                \hfill\ (\rlabel{U3})                                                          \\
                (\cheadd,\cheadd) & \melif \cheadd\in M\land\cpath\notin M
                \hfill\ (\rlabel{U4})                                                          \\
                (\cheadd,\cstay)  & \moth \hfill\ (\rlabel{U5})
            \end{cases} \\
            \xius(\cheadd,M) & \!\!=
            \begin{cases}
                (\cfin,\cinit)    & \mif \cinit,\cneigh\in M
                \land\cpath,\cheadd\notin M \hfill\ (\rlabel{U6}) \\
                (\cheadd,\cneigh) & \melif \cheadd,\cneigh\in M
                \hfill\ (\rlabel{U7})                             \\
                (\cpath,\cheadd)  & \melif \cheadd,\cpath\in M
                \hfill\ (\rlabel{U8})                             \\
                (\cfin,\cheadd)   & \melif \cheadd\in M
                \hfill\ (\rlabel{U9})                             \\
                (\cheadd,\cinit)  & \melif \cpath,\cinit\in M
                \hfill\ (\rlabel{U10})                            \\
                (\cheadd,\cpath)  & \melif \cpath\in M
                \hfill\ (\rlabel{U11})                            \\
                (\cpath,\cinit)   & \melif \cinit\in M
                \hfill\ (\rlabel{U12})                            \\
                (\cfin,\cstop)    & \moth \hfill\ (\rlabel{U13})
            \end{cases}                              \\
            \xius(\cpath,M)  & \!\!=
            \begin{cases}
                (\cinit,\cheadd)  & \mif \cneigh,\cheadd\in M
                \hfill\ (\rlabel{U14})                           \\
                (\cpath,\cneigh)  & \melif \cneigh\in M
                \hfill\ (\rlabel{U15})                           \\
                (\cheadd,\cstay)  & \melif \cheadd\notin M
                \hfill\ (\rlabel{U16})                           \\
                (\cheadd,\cheadd) & \moth \hfill\ (\rlabel{U17})
            \end{cases}                               \\
            \xius(\cneigh,M) & \!\!=
            \begin{cases}
                (\cinit,\cpath)  & \mif \cpath\in M\land\cheadd\notin M
                \hfill\ (\rlabel{U18})                                  \\
                (\cinit,\cheadd) & \moth \hfill\ (\rlabel{U19})
            \end{cases}
        \end{align*}
    }
\end{algorithm}
\endgroup

Algorithm~\ref{alg:uniform-five} contains the thirteen rules of the five-color algorithm for triangle-free graphs by Takahashi et al.~\cite{takahashi2026recolorable} as \rlabel{U1}, \rlabels{U3--U5}, \rlabels{U7--U13}, \rlabel{U17}, and \rlabel{U19}.
These rules check the $\cinit$-colored neighbors of the current path endpoint one at a time to determine whether it can extend the path.
They temporarily mark rejected candidates with $\cneigh$ and restore them to $\cinit$ before completing an expansion or a backtracking step, so that they can be tested again.
We call the removal of these marks \emph{cleanup}.
On triangle-free graphs, no rejected candidate is adjacent to the predecessor of the endpoint, so every mark is restored through the endpoint.
When the endpoint and its predecessor lie in a clique block, every rejected candidate is adjacent to the predecessor (Lemma~\ref{lem:ktf-uniform}).
The six new rules \rlabel{U2}, \rlabel{U6}, \rlabels{U14--U16}, and \rlabel{U18} handle this case.
They temporarily recolor the predecessor with $\cinit$, finish the endpoint, and restore the marks through the predecessor.
No rule requests $\cfin$ as a destination color, so the rule table omits the case in which the current vertex is colored $\cfin$.

\subsection{Semi-DFS and the two cleanup conditions}

A semi-DFS state consists of a path $P=(u_1,u_2,\dots,u_k)$ with $u_1=s$ and a finished set $F\subseteq V$.
The initial state is $((s),\emptyset)$.
We define $U(P,F)=\upf$ as the set of neighbors of $u_k$ that are neither finished nor on $P$ and are adjacent to no earlier vertex of $P$.
If $U(P,F)\ne\emptyset$, semi-DFS appends an arbitrary vertex of $U(P,F)$ to $P$.
This is a semi-DFS \emph{expansion}.
If $U(P,F)=\emptyset$ and $k\ge2$, it deletes $u_k$ from $P$ and adds it to $F$.
This is semi-DFS \emph{backtracking}.
If $P=(s)$ and $U(P,F)=\emptyset$, it adds $s$ to $F$ and stops.

\begin{jlemma}[Semi-DFS~\cite{takahashi2026recolorable}]
    \label{lem:semidfs}
    The path $P$ remains induced throughout semi-DFS.
    After at most $2n$ iterations, the procedure has visited every vertex and terminates at $s$.
\end{jlemma}

To formulate the two conditions, we call a semi-DFS state \emph{reachable} if arbitrary candidate choices can produce it from the initial state.
For such a state with $P=(u_1,\dots,u_k)$, we write $a=u_k$ for the head and $p=u_{k-1}$ for its predecessor when $k\ge2$.
We call any $b\in U(P,F)$ an eligible candidate and define the rejected set as
\[
    R(P,F)=\{x\in N(a)\setminus(F\cup V(P)) :
    N(x)\cap\{u_1,\ldots,u_{k-1}\}\neq\emptyset\}.
\]
We call any $x\in R(P,F)$ a rejected candidate.
The simulation represents $(P,F)$ by coloring $V\setminus(V(P)\cup F)$ with $\cinit$, $F$ with $\cfin$, $u_1,\dots,u_{k-1}$ with $\cpath$, and $u_k$ with $\cheadd$.

With these color roles, two moves during cleanup require structural conditions.
During expansion cleanup, the old and new endpoints $a,b$ both have color $\cheadd$ as the $\cneigh$ marks are restored.
If a rejected candidate $x$ were adjacent to $b$, the adversary could send the agent from $x$ to the wrong endpoint.
After a backtracking scan, either no rejected candidate is adjacent to $p$ and all marks are restored through $a$, or every rejected candidate must be adjacent to $p$ and to no earlier path vertex.
Otherwise the adversary could choose the wrong path vertex.

These two possible failures lead to the following conditions on every reachable state $(P,F)$.
\begin{enumerate}[(i)]
    \item \emph{Expansion separation}: for all $b\in U(P,F)$ and $x\in R(P,F)$, $\{b,x\}\notin E$.
    \item \emph{Backtracking uniformity}: if $k\ge2$, $p=u_{k-1}$, and $U(P,F)=\emptyset$, then either $R(P,F)\cap N(p)=\emptyset$, or every $x\in R(P,F)$ satisfies $\{x,p\}\in E$ and $N(x)\cap V(P)=\{p,a\}$.
\end{enumerate}
A graph has the \emph{uniform separation property} when both conditions hold for every reachable state from every start $s$.
Figure~\ref{fig:uniform-separation} shows the two configurations.

\begin{figure}[htbp]
    \centering
    \definecolor{semipath}{RGB}{44,101,154}
    \definecolor{semihead}{RGB}{201,91,38}
    \definecolor{semineigh}{RGB}{0,130,120}
    \begin{tikzpicture}[
            x=1cm,y=1cm,font=\small,
            pathv/.style={circle,draw=semipath,thick,fill=semipath!12,minimum size=5.6mm,inner sep=0pt},
            headv/.style={circle,draw=semihead,thick,fill=semihead!14,minimum size=5.6mm,inner sep=0pt},
            rej/.style={rectangle,draw=semineigh,thick,fill=semineigh!14,minimum size=5.6mm,inner sep=0pt},
            sel/.style={headv,very thick,minimum size=5.8mm},
            pathbox/.style={rounded corners=2pt,draw=semipath!45,fill=semipath!4},
            setbox/.style={rounded corners=3pt,draw=semineigh!70,densely dashed},
            ptitle/.style={font=\bfseries}
        ]
        \begin{scope}
            \node[ptitle] at (2.15,2.05) {(a) Expansion separation};
            \draw[pathbox] (-0.35,-0.36) rectangle (3.35,0.36);
            \node[anchor=east] at (-0.43,0) {$P$};
            \node[pathv] (ui) at (0,0) {$u_i$};
            \node[pathv] (p)  at (1.6,0) {$p$};
            \node[headv] (a)  at (3.0,0) {$a$};
            \draw[thick] (ui)--node[below] {$\cdots$} (p)--(a);
            \node[rej] (x) at (0.65,1.15) {$x$};
            \node[sel] (b) at (4.30,1.15) {$b$};
            \draw[thick] (x)--(ui) (x)--(a) (b)--(a);
            \draw[densely dashed,draw=black!55] (x)--(b);
            \fill[white] (2.475,1.15) circle (0.16);
            \draw[black!70,line width=0.8pt]
            (2.375,1.05)--(2.575,1.25) (2.375,1.25)--(2.575,1.05);
            \node[anchor=north,align=center] at (1.95,-0.62)
            {$x\in R(P,F),\ b\in U(P,F)$\\$\Longrightarrow\ \{x,b\}\notin E$};
        \end{scope}
        \draw[black!20,line width=0.5pt] (5.00,-0.45)--(5.00,2.25);
        \begin{scope}[xshift=6.05cm]
            \node[ptitle] at (2.15,2.05) {(b) Backtracking uniformity};
            \draw[pathbox] (-0.35,-0.36) rectangle (3.55,0.36);
            \node[anchor=east] at (-0.43,0) {$P$};
            \node[pathv] (u)  at (0,0) {$u_i$};
            \node[pathv] (rp) at (1.45,0) {$p$};
            \node[headv] (ra) at (3.15,0) {$a$};
            \draw[thick] (u)--node[below] {$\cdots$} (rp)--(ra);
            \draw[setbox] (0.72,0.82) rectangle (3.88,1.68);
            \node[fill=white,inner sep=1pt] at (2.30,1.68) {$R(P,F)$};
            \node[rej] (rx) at (1.25,1.22) {$x$};
            \node at (2.30,1.22) {$\cdots$};
            \node[rej] (rz) at (3.35,1.22) {$z$};
            \draw[thick] (rx)--(rp) (rx)--(ra) (rz)--(rp) (rz)--(ra);
            \node[anchor=north,align=center] at (2.05,-0.62)
            {$R(P,F)\cap N(p)\ne\emptyset$\\$\Longrightarrow\ N(y)\cap V(P)=\{p,a\}\ \ (\forall y\in R(P,F))$};
        \end{scope}
    \end{tikzpicture}
    \caption{The two conditions that prevent cleanup from moving to the wrong endpoint.
        Blue, orange, and teal denote $\cpath$, $\cheadd$, and $\cneigh$. The crossed dashed edge in (a) is forbidden.}
    \label{fig:uniform-separation}
\end{figure}
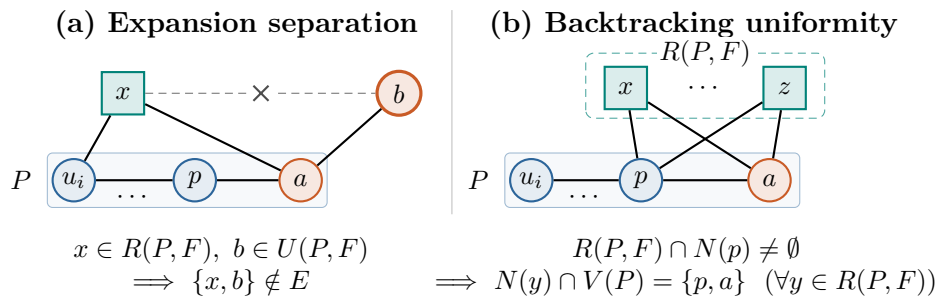

\subsection{Five-color simulation}

A configuration $\gamma=(v,\psi)$ is \emph{consistent} if the vertices colored neither $\cinit$ nor $\cfin$ form an induced path $P(\gamma)=(u_1,\dots,u_k)$ with $u_1=s$, every $u_i$ with $i<k$ is colored $\cpath$, and $u_k=v$ is colored $\cheadd$.
This path is unique.
With $F(\gamma)=\{v\in V:\psi(v)=\cfin\}$, the configuration represents the semi-DFS state $(P(\gamma),F(\gamma))$.

\begin{jlemma}
    \label{lem:uniform-simulation}
    On a graph with the uniform separation property, from every consistent configuration representing a reachable semi-DFS state $(P,F)$, Algorithm~\ref{alg:uniform-five} reaches one of the following outcomes in finite time.
    \begin{enumerate}[(i)]
        \item If $U(P,F)\ne\emptyset$, it reaches a consistent configuration representing an expansion $P\mapsto P\circ b$ for some $b\in U(P,F)$.
        \item If $U(P,F)=\emptyset$ and $|P|\ge2$, it reaches the consistent configuration representing the corresponding backtracking step.
        \item If $P=(s)$ and $U(P,F)=\emptyset$, it stops at $s$.
    \end{enumerate}
\end{jlemma}

\begin{proof}
    \emph{Start.} For $P=(s)$, \rlabel{U12} followed by \rlabel{U5} makes an $\cinit$-colored neighbor the new head, and \rlabel{U13} stops if none exists.

    \emph{Scan.} For $|P|\ge2$, \rlabel{U10} visits the $\cinit$-colored neighbors of the head $a$ one at a time, and each visit returns through the unique $\cheadd$ vertex $a$.
    \rlabel{U3} marks a candidate $\cneigh$ when the agent sees $\cpath$ there, and otherwise \rlabel{U4} colors it $\cheadd$.
    The scan ends either with a new head $b$ and the scanned rejected candidates marked, or with every $\cinit$-colored neighbor of $a$ marked.

    \emph{Expansion cleanup.} With $a$ and $b$ both $\cheadd$, \rlabel{U7} moves to a mark and \rlabel{U19} restores it and returns to $\cheadd$, which is $a$ alone because expansion separation keeps every rejected candidate away from $b$.
    After the last mark, \rlabel{U8} recolors $a$ with $\cpath$ and moves to $b$, which gives the consistent configuration of $P\circ b$.

    \emph{Backtracking.} If no candidate is eligible, \rlabel{U11} moves to the predecessor $p$, and backtracking uniformity leaves two patterns.
    If $p$ is adjacent to no mark, \rlabel{U17} makes $p$ a head and returns to $a$, repetitions of \rlabels{U7 and U19} restore the marks through $a$, and \rlabel{U9} finishes $a$ and moves to $p$.
    If $p$ is adjacent to one mark, it is adjacent to all of them, so \rlabel{U14} temporarily recolors $p$ with $\cinit$, \rlabel{U6} finishes $a$ and moves to $p$ as the unique $\cinit$ neighbor of $a$, and \rlabel{U2} recolors $p$ with $\cpath$.
    Repetitions of \rlabels{U15 and U18} then restore the marks through $p$, which is their unique $\cpath$ neighbor by uniformity, and \rlabel{U16} makes $p$ the head.
    Both patterns give the consistent configuration of $P$ without $a$.
    Each scan or cleanup iteration removes one remaining candidate, so every phase is finite.
\end{proof}

\begin{jtheorem}
    \label{thm:uniform}
    Every graph with the uniform separation property can be explored by Algorithm~\ref{alg:uniform-five} using five colors and $O(n\Delta)$ moves.
\end{jtheorem}

\begin{proof}
    At a nonisolated start, \rlabel{U1} followed by \rlabel{U5} creates the first consistent configuration.
    At an isolated start, \rlabel{U5} followed by \rlabel{U13} stops there.
    The simulation established by Lemma~\ref{lem:uniform-simulation} follows semi-DFS step by step.
    Lemma~\ref{lem:semidfs} therefore gives finite exploration of all vertices and termination at $s$.
    There are at most $2n$ semi-DFS steps.
    Each examines at most $\Delta$ candidates and restores their colors with constant work per candidate, for $O(n\Delta)$ moves.
\end{proof}

\subsection{Structural characterization}

\begin{jlemma}
    \label{lem:ktf-uniform}
    Expansion separation holds in every reachable semi-DFS state from every start if and only if the graph belongs to $\gktf$.
    Every graph in $\gktf$ also satisfies backtracking uniformity.
    Consequently, the uniform separation property holds exactly for graphs in $\gktf$.
\end{jlemma}

\begin{proof}
    We first suppose that $G\in\gktf$ and that expansion separation fails for an eligible candidate $b$ and a rejected candidate $x$.
    Since both are adjacent to the head $a$ and $\{b,x\}\in E$, the vertices $a,b,x$ form a triangle.
    Because $x$ is rejected, it has a neighbor $u_i$ earlier on $P$.
    The edge $\{x,u_i\}$, the subpath from $u_i$ to $a$, and $\{a,x\}$ form a cycle sharing the edge $\{a,x\}$ with that triangle.
    The triangle and cycle therefore lie in one block, which must be a clique because $G\in\gktf$.
    This forces $\{b,u_i\}\in E$, contradicting the eligibility of $b$.

    Next, we assume $x\in R(P,F)\cap N(p)$.
    The triangle $x,p,a$ lies in a block containing $\{p,a\}$.
    Because $G\in\gktf$, this block is a clique.
    For every rejected candidate $y$, its edge to an earlier path vertex together with the subpath to $a$ and the edge $\{a,y\}$ forms a cycle using $\{p,a\}$.
    Hence this cycle lies in the same clique block, which forces $\{y,p\}\in E$.
    If $y$ were also adjacent to some $u_j$ with $j\le k-2$, the same clique block would force $\{u_j,a\}\in E$, contradicting the fact that $P$ is induced.
    Thus backtracking uniformity follows.

    For the converse, we assume that expansion separation always holds and that a block $B$ contains a triangle but is not a clique.
    We will construct a reachable state with adjacent vertices $b\in U(P,F)$ and $z\in R(P,F)$, contradicting expansion separation.
    We choose an inclusion-maximal clique $K\subsetneq V(B)$ that contains a triangle.
    By biconnectivity, some component $C$ of $B-K$ has two distinct neighbors in $K$.
    We choose a shortest path $Q=a,q_1,\dots,q_t,z$ with $t\ge1$, all internal vertices in $C$, and distinct endpoints in $K$.
    Then $Q-z$ is induced.
    Indeed, any chord of $Q-z$ would give a shorter path with two endpoints in $K$ and all internal vertices in $C$.
    We claim that some $b\in K\setminus\{a,z\}$ is adjacent to no $q_i$.
    If $t\ge2$, any $b\in K\setminus\{a,z\}$ suffices, and such a vertex exists because $K$ contains a triangle.
    An edge $\{b,q_i\}$ would contradict the minimality of $Q$.
    For $i<t$, the path $a,q_1,\dots,q_i,b$ is shorter than $Q$.
    For $i=t$, the path $b,q_t,z$ is shorter than $Q$.
    If $t=1$, maximality of $K$ gives a vertex $b\in K\setminus\{a,z\}$ nonadjacent to $q_1$, since otherwise $K\cup\{q_1\}$ would be a larger clique.
    We start semi-DFS at $q_t$, and the adversary follows the vertices of $Q-z$ in reverse order until it reaches $a$.
    These choices are eligible because $Q-z$ is induced.
    The resulting path is the reverse of $Q-z$ and has $F=\emptyset$.
    The vertex $z$ is rejected because it is adjacent to the head $a$ and the earlier path vertex $q_t$.
    The vertex $b$ is eligible because it is adjacent to $a$ and to none of the $q_i$.
    Finally, $\{b,z\}\in E$ because both vertices lie in the clique $K$, contradicting expansion separation.
\end{proof}

\begin{jtheorem}
    \label{thm:main}
    Every graph in $\gktf$ can be explored by Algorithm~\ref{alg:uniform-five} using five colors and $O(n\Delta)$ moves.
\end{jtheorem}

\begin{proof}
    Lemma~\ref{lem:ktf-uniform} gives the uniform separation property, so Theorem~\ref{thm:uniform} applies.
\end{proof}

The diamond $K_4-e$ is $\varphi$-free but not in $\gktf$, whereas $K_5$ is in $\gktf$ but not $\varphi$-free, so the two classes are incomparable.

A graph is \emph{unichord-free} if no cycle has exactly one chord~\cite{trotignon2010structure}.
A structure theorem of Trotignon and Vu\v{s}kovi\'c gives a further consequence for this class.

\begin{jcorollary}
    \label{cor:unichord}
    Every unichord-free graph can be explored with five colors and $O(n\Delta)$ moves.
\end{jcorollary}

\begin{proof}
    Every block is an induced subgraph and hence unichord-free, and it has no cut vertex.
    By \cite[Theorem~2.2]{trotignon2010structure}, a connected unichord-free graph that contains a triangle is a clique or has a cut vertex in the maximal clique containing that triangle.
    Hence every block containing a triangle is a clique, so the graph belongs to $\gktf$ and Theorem~\ref{thm:main} applies.
\end{proof}

\section{A Four-Color Lower Bound for Subcubic Pseudotrees}
\label{sec:pseudotree-lower}

In this section, we use the obstruction family $\mathcal H$ of nine graphs in Figure~\ref{fig:cactus-lower-witness} to show that no single rule with at most three colors explores every subcubic pseudotree.

By adding unused colors if necessary and renaming colors, we may assume w.l.o.g.~that $C=\{0,1,2\}$, with $0$ as the initial color.
By Lemma~\ref{lem:stay-elimination}, candidate successful rules may be restricted to those that never stay at a reachable observation.
Initially, the current vertex and all its neighbors have color $0$, so the observation depends only on the degree of the starting vertex.
At a start of degree $d\in\{1,2,3\}$, stopping would leave vertices unvisited, and color $0$ is the only neighbor color that can be requested.
Thus the first action writes a color $r_d\in C$ determined by $d$ and moves to a neighbor of color $0$.
We call the triple $(r_1,r_2,r_3)$ the \emph{initial tuple} and abbreviate it as $\mathtt{r_1r_2r_3}$.

\begin{figure}[htbp]
    \centering
    \begin{tikzpicture}[
            x=1.2cm,y=1.2cm,
            hv/.style={circle,fill=black,inner sep=1.6pt},
            he/.style={line width=0.6pt},
            hname/.style={font=\small,anchor=north}
        ]

        \begin{scope}[shift={(2.70,1.60)}]
            \node[hv](a)at(-0.42,-0.24){}; \node[hv](b)at(0.42,-0.24){};
            \node[hv](c)at(0,0.28){};
            \draw[he](a)--(b)--(c)--(a);
            \node[hname]at(0,-0.58){$C_3$};
        \end{scope}

        \begin{scope}[shift={(5.20,1.60)}]
            \node[hv](a)at(0,-0.05){}; \node[hv](b)at(-0.48,-0.32){};
            \node[hv](c)at(0.48,-0.32){}; \node[hv](d)at(0,0.40){};
            \draw[he](a)--(b) (a)--(c) (a)--(d);
            \node[hname]at(0,-0.58){$K_{1,3}$};
        \end{scope}
        \begin{scope}[shift={(7.70,1.60)}]
            \node[hv](a)at(-0.36,-0.28){}; \node[hv](b)at(0.36,-0.28){};
            \node[hv](c)at(0.36,0.28){}; \node[hv](d)at(-0.36,0.28){};
            \draw[he](a)--(b)--(c)--(d)--(a);
            \node[hname]at(0,-0.58){$C_4$};
        \end{scope}
        \begin{scope}[shift={(10.20,1.60)}]
            \node[hv](a)at(-0.40,-0.30){}; \node[hv](b)at(0.40,-0.30){};
            \node[hv](c)at(0,0.14){}; \node[hv](d)at(0,0.52){};
            \draw[he](a)--(b)--(c)--(a) (c)--(d);
            \node[hname]at(0,-0.58){$C_3$ + leaf};
        \end{scope}

        \begin{scope}[shift={(1.45,-0.10)}]
            \node[hv](c)at(-0.62,0){}; \node[hv](b)at(-0.22,0){};
            \node[hv](a)at(0.18,0){}; \node[hv](d)at(0.56,0.28){};
            \node[hv](e)at(0.56,-0.28){};
            \draw[he](c)--(b)--(a) (a)--(d) (a)--(e);
            \node[hname]at(0,-0.70){$T_5$};
        \end{scope}
        \begin{scope}[shift={(3.95,-0.10)}]
            \foreach \i/\n in {90/n1,162/n2,234/n3,306/n4,18/n5}{
                    \node[hv](\n)at(\i:0.42){};}
            \draw[he](n1)--(n2)--(n3)--(n4)--(n5)--(n1);
            \node[hname]at(0,-0.70){$C_5$};
        \end{scope}
        \begin{scope}[shift={(6.45,-0.10)}]
            \node[hv](a)at(-0.30,-0.30){}; \node[hv](b)at(0.30,-0.30){};
            \node[hv](c)at(0.30,0.24){}; \node[hv](d)at(-0.30,0.24){};
            \node[hv](e)at(-0.62,0.52){};
            \draw[he](a)--(b)--(c)--(d)--(a) (d)--(e);
            \node[hname]at(0,-0.70){$C_4$ + leaf};
        \end{scope}
        \begin{scope}[shift={(8.95,-0.10)}]
            \node[hv](a)at(-0.60,-0.26){}; \node[hv](b)at(-0.60,0.26){};
            \node[hv](c)at(-0.20,0){}; \node[hv](d)at(0.22,0){};
            \node[hv](e)at(0.62,0){};
            \draw[he](a)--(b)--(c)--(a) (c)--(d)--(e);
            \node[hname]at(0,-0.70){$C_3$ + tail of length 2};
        \end{scope}
        \begin{scope}[shift={(11.45,-0.10)}]
            \node[hv](a)at(0,0.34){}; \node[hv](b)at(-0.34,-0.16){};
            \node[hv](c)at(0.34,-0.16){}; \node[hv](d)at(-0.66,-0.46){};
            \node[hv](e)at(0.66,-0.46){};
            \draw[he](a)--(b)--(c)--(a) (b)--(d) (c)--(e);
            \node[hname]at(0,-0.70){$C_3$ + 2 leaves};
        \end{scope}
    \end{tikzpicture}
    \caption{The obstruction family $\mathcal H$ of nine subcubic pseudotrees on at most five vertices.}
    \label{fig:cactus-lower-witness}
\end{figure}
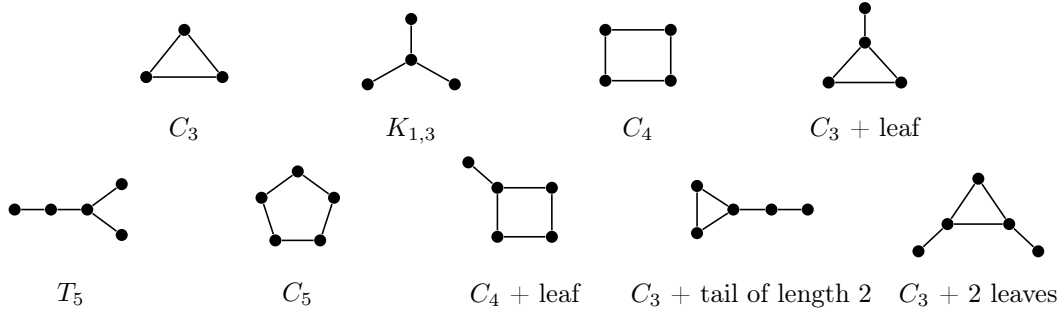

\begin{jtheorem}
    \label{thm:subcubic-pseudotree-lower}
    No single rule with at most three colors explores every subcubic pseudotree.
\end{jtheorem}

\begin{proof}
    Suppose that a common successful rule exists, and let $\mathtt{r_1r_2r_3}$ be its initial tuple.
    The graphs $C_3$ and $K_{1,3}$ are subcubic pseudotrees containing vertices of degrees $2$ and $1,3$, respectively.
    If $r_d=0$ for some $d\in\{1,2,3\}$, the first move from a degree-$d$ vertex in one of these graphs reaches the initial configuration for its destination.
    Lemma~\ref{lem:execution-merging} then contradicts success, so $r_1,r_2,r_3\ne0$.
    Up to exchanging colors $1$ and $2$, the remaining eight tuples reduce to the four representatives $\itup111,\itup121,\itup211,\itup221$.

    After this reduction, one rule must still explore every graph in the obstruction family $\mathcal H$ of Figure~\ref{fig:cactus-lower-witness}.
    The computer-assisted Lemma~\ref{lem:four-initial-obstruction} in Appendix~\ref{app:cactus-three-certificate} excludes such a common rule for each of the four initial representatives.
    Hence no single rule with at most three colors explores every subcubic pseudotree.
\end{proof}

Consequently, every graph class that contains all subcubic pseudotrees and is contained in $\gcb$ has color complexity exactly four, since Algorithm~\ref{alg:cactus4} explores every graph in $\gcb$.

\section{Conclusion}
\label{sec:conclusion}

We obtain tight three-color bounds for trees and simple cycles and tight four-color bounds for every class between subcubic pseudotrees and $\gcb$, including cacti.
The five-color cleanup conditions characterize $\gktf$.
A four-color lower bound without computer assistance or for a single graph remains open.

\section*{Use of AI tools}
OpenAI Codex (GPT-5.6 Sol/\texttt{xhigh}) assisted with editing and verification and implemented most of the lower-bound code, all of which the authors reviewed.

\bibliography{references}

\appendix

\section{Certificate for the Three-Color Subcubic-Pseudotree Lower Bound}
\label{app:cactus-three-certificate}

This appendix proves the computer-assisted Lemma~\ref{lem:four-initial-obstruction} used in Section~\ref{sec:pseudotree-lower}.
We first specify the graphs and initial tuples and state the lemma.
We then prove that the base encoding represents a common exploration rule and that the additional clauses are sound.
Finally, we describe the checked unsatisfiability certificates, reproduction procedure, and trusted software, and complete the proof of the lemma.

The source code, fixed inputs, and verification scripts for this certificate are available in the repository below.
\begin{center}
    \url{https://github.com/littlegirl0820/color-complexity-recolorable-exploration}
\end{center}

\subsection{Obstruction family and initial representatives}

Throughout this appendix, $C=\{0,1,2\}$ and the initial color is $0$.
At an all-zero observation of degree $d\in\{1,2,3\}$, any successful rule that never stays at a reachable observation must write some $r_d\in C$ and move to color $0$.
We call $\mathtt{r_1r_2r_3}$ its initial tuple.

The family $\mathcal H$ in Figure~\ref{fig:cactus-lower-witness} consists of two trees and seven graphs containing exactly one cycle, all on three to five vertices.
The graph $T_5$ is $K_{1,3}$ with one edge subdivided.

Throughout the encoding, $\tau$ ranges over the four representatives $\itup111,\itup121,\itup211,\itup221$.

\begin{jlemma}[Computer-assisted]
    \label{lem:four-initial-obstruction}
    Let $C=\{0,1,2\}$ be the color set, with initial color $0$.
    Let
    \[
        \tau\in\{\itup111,\itup121,\itup211,\itup221\}.
    \]
    No single deterministic three-color rule with initial tuple $\tau$ that never stays at a reachable observation explores every graph in $\mathcal H$ from every start and against every adversarial choice.
\end{jlemma}

We prove this lemma at the end of the appendix using the encoding and verification results below.

\subsection{Exact base encoding}
\label{app:cactus-base-encoding}

We fix one representative $\tau$.
The possible observations at nonisolated vertices of graphs of maximum degree three form the set
\[
    \mathcal O=\{(c,M)\in C\times\calM(C):1\le |M|\le3\}.
\]
Since the number of multisets of size $d$ over three colors is $\binom{d+2}{2}$,
\[
    |\mathcal O|=3\sum_{d=1}^3\binom{d+2}{2}=3(3+6+10)=57.
\]
The encoding represents the rule using the ten actions in $\mathsf A=\{a_{\mathrm{stop}}\}\cup C^2$.
An action $(c',d)\in C^2$ recolors the current vertex with $c'$ and moves to a neighbor of color $d$.
For an observation $o=(c,M)$, the available actions are
\[
    \mathsf{Act}(o)=\{a_{\mathrm{stop}}\}\cup\{(c',d)\in C^2:d\in M\}.
\]
A stopping action ends the execution immediately, so its written color is irrelevant.
We therefore represent all three actions $(c',\cstop)$, $c'\in C$, by the single action $a_{\mathrm{stop}}$.

For every $o\in\mathcal O$ and $a\in\mathsf A$, the generator allocates a rule variable $x_{o,a}$ and adds the exactly-one clauses
\[
    \bigvee_{a\in\mathsf A}x_{o,a},
    \qquad
    \neg x_{o,a}\lor\neg x_{o,b}
    \quad(a,b\in\mathsf A,\ a\ne b).
\]
For each action $a\in\mathsf A\setminus\mathsf{Act}(o)$ that requests a color absent from the neighborhood, the generator adds the unit clause $\neg x_{o,a}$.
Three further unit clauses fix the actions at the all-zero observations of degrees 1, 2, and 3 according to $\tau$.
We write $P_{\mathrm{rule}}^\tau$ for these exactly-one, availability, and initial-tuple clauses.

The remaining variables and clauses are generated separately for every $H\in\mathcal H$.
We denote the set of all configurations of $H$ by $Q_H$.
For $q=(v,\psi)\in Q_H$, we write $M_\psi(v)=\{\!\{\psi(u):u\in N(v)\}\!\}$ and $o(q)=(\psi(v),M_\psi(v))$.
For a configuration $q=(v,\psi)$ and a selected action $a$, the set of configurations that the adversary may choose next is denoted by $\delta_H(q,a)$.
If $a$ is not stopping, then $\delta_H(q,a)$ is nonempty by the definition of allowed actions.

The encoding expresses four requirements.
It must include every configuration that the adversary can reach, ensure finite termination, require return to the chosen start, and require every vertex to be visited before stopping.
Variables $R_q$ and ranks $h(q)$ express the first two requirements.
Variables $T_{s,q}$ remember the start for the third.
Variables $R_q^{\neg m}$ track an execution that has not visited $m$ for the fourth.

The encoding introduces a reachability variable $R_q$ for each configuration $q$.
For every start $s$, the symbol $q_s^0$ denotes the all-zero initial configuration, and the encoding contains the unit clause $R_{q_s^0}$.
For every $q'\in\delta_H(q,a)$, it contains
\[
    \neg R_q\lor\neg x_{o(q),a}\lor R_{q'}
\]
to require $q'$ to be reachable whenever $q$ is reachable and the rule selects $a$.
We write $P_{\mathrm{reach}}(H)$ for these clauses.

Each $q$ has a rank $h(q)$ of $\lceil\log_2|Q_H|\rceil$ bits.
For every nonstopping successor $q'\in\delta_H(q,a)$, the encoding contains a standard binary-comparator CNF for
\[
    R_q\land x_{o(q),a}\quad\Longrightarrow\quad h(q')<h(q).
\]
This condition forces the rank to decrease at every reachable transition.
The ranks of configurations marked unreachable may have arbitrary values.
We write $P_{\mathrm{term}}(H)$ for the comparator clauses.

The variables $R_q$ combine executions from every start, so a separate family records where each execution began.
For the start-specific reachability variables $T_{s,q}$, the encoding contains the unit clause $T_{s,q_s^0}$.
For every $q'\in\delta_H(q,a)$, it contains
\[
    \neg T_{s,q}\lor\neg x_{o(q),a}\lor T_{s,q'}
\]
to propagate reachability from start $s$ to every adversarial successor.
With $a_{\mathrm{stop}}$ denoting the stopping action, the encoding also contains
\[
    \neg T_{s,q}\lor\neg x_{o(q),a_{\mathrm{stop}}}
\]
whenever the current position in $q$ is not $s$.
This prevents an execution started at $s$ from stopping elsewhere.
We write $P_{\mathrm{return}}(H)$ for these clauses.

The variable $R^{\neg m}_q$ means that $q$ is reachable along an execution prefix that has not yet visited vertex $m$.
For every start $s$ and vertex $m\ne s$, the encoding contains the unit clause $R^{\neg m}_{q_s^0}$.
It propagates $R^{\neg m}_q$ over a transition $q\to q'$ only when the current position of $q'$ is not $m$.
For every such $q'\in\delta_H(q,a)$, it contains
\[
    \neg R^{\neg m}_q\lor\neg x_{o(q),a}\lor R^{\neg m}_{q'}.
\]
For every $q$ and $m$, it also contains
\[
    \neg R^{\neg m}_q\lor\neg x_{o(q),a_{\mathrm{stop}}}.
\]
Hence an execution leaving any vertex unvisited cannot stop.
We write $P_{\mathrm{visit}}(H)$ for these clauses.

The encoding does not require each configuration marked reachable to have a marked predecessor.
When constructing an assignment from a successful rule, we use exactly the configurations that the rule can reach.
In the other direction, the initial unit clauses and the transition clauses ensure that every configuration actually reached by the rule is marked reachable.

The base encoding is
\[
    B_\tau
    =P_{\mathrm{rule}}^\tau
    \land\bigwedge_{H\in\mathcal H}
    \bigl(
    P_{\mathrm{reach}}(H)
    \land P_{\mathrm{term}}(H)
    \land P_{\mathrm{return}}(H)
    \land P_{\mathrm{visit}}(H)
    \bigr).
\]

\subsection{Exactness of the base encoding}

\begin{jlemma}
    \label{lem:sat-base-exactness}
    The formula $B_\tau$ is satisfiable if and only if there exists a three-color rule with initial tuple $\tau$ satisfying two conditions.
    The rule never stays at a reachable observation.
    It explores every graph in $\mathcal H$ from every start against every adversarial choice.
\end{jlemma}

\begin{proof}
    A satisfying assignment selects a unique available action at each observation.
    These actions define the rule.
    The $R_q$ clauses mark every configuration that the adversary may reach from a marked configuration.
    The ranks strictly decrease at every such move, so an execution cannot be infinite.
    Every maximal execution therefore stops.
    The $T_{s,q}$ clauses require every stop to be at its start.
    The $R^{\neg m}_q$ clauses forbid a stop before every vertex has been visited.

    For the converse, we start with a successful rule.
    We first redefine its outputs at unreachable observations to use arbitrary available non-stay actions.
    This does not change any execution.
    We assign each reachability variable according to whether the corresponding configuration is actually reachable, including the variants that record the start and an unvisited vertex.
    A directed cycle in the reachable configuration-transition graph would let the adversary repeat it forever, so that graph is acyclic.
    We order its configurations topologically so that every transition goes from a larger position to a smaller one.
    These positions give the required ranks, and all clauses of $B_\tau$ follow.
\end{proof}

\subsection{Additional clauses checked independently}

The exact encoding $B_\tau$ is large and difficult to refute directly.
We therefore add three families of clauses that every successful rule satisfies.
They shorten the propositional proof but do not change the represented exploration problem.

For distinct starts $s,t$ and every configuration $q$, the encoding contains
\[
    \neg T_{s,q}\lor\neg T_{t,q}.
\]
These clauses form $P_{\mathrm{merge}}$.
If a successful rule is assigned its actual start-specific reachable sets, both variables cannot be true.
Otherwise the two executions reach $q$, and Lemma~\ref{lem:execution-merging} gives a failing maximal execution, contradicting the assumption that the rule explores the graph.

Separate checkers verify the remaining two kinds of additional clauses.
The generator imports 279,197 such clauses over the rule variables.
They consist of 252,515 trace clauses and 26,682 domain clauses.
A trace clause has the form
\[
    \bigvee_i \neg x_{o_i,a_i}.
\]
The simultaneous assignments $x_{o_i,a_i}=1$ fix a partial rule with $\xi(o_i)=a_i$.
With these actions fixed, an execution fails if it stops away from its start, stops leaving some vertex unvisited, requests an absent color, or runs forever.
The supplied checker searches for a failing execution prefix or a reachable cycle using only the fixed actions on some $H\in\mathcal H$ and start $s$.
It stops a search branch when the next observation has no fixed action.
Every accepted witness therefore remains a failure for every complete rule containing the fixed actions.
Every successful rule must therefore satisfy the clause.

A domain certificate assigns a nonempty set $D(o)\subseteq\mathsf{Act}(o)$ to every observation and contributes
\[
    \bigvee_o\ \bigvee_{a\in\mathsf{Act}(o)\setminus D(o)} x_{o,a}.
\]
The checker solves a finite reachability game.
For a fixed graph $H$ and start $s$, a game state records the agent position, the complete vertex coloring, and the set of vertices already visited.
The controller may choose an action from $D(o)$ separately at every game state, even when two game states have the same observation $o$.
The checker accepts only if the all-zero game state at start $s$ is losing against the adversary.
This controller can make more choices than a deterministic rule restricted to $D$, because it may choose different actions at states with the same observation.
If even this controller loses when restricted to $D$, then every deterministic rule using only actions in $D$ also loses.
Every successful rule must therefore select an action outside $D$ at some observation and satisfy the clause.
We write $P_{\mathrm{trace}}$ and $P_{\mathrm{domain}}$ for these two clause families and set
\[
    F_\tau=B_\tau\land P_{\mathrm{merge}}\land P_{\mathrm{trace}}\land P_{\mathrm{domain}}.
\]
Lemma~\ref{lem:execution-merging} proves the soundness of $P_{\mathrm{merge}}$.
Separate routines verify $P_{\mathrm{trace}}$ and $P_{\mathrm{domain}}$ independently of the SAT solver and DRAT proof.
Neither check assumes a fixed initial representative $\tau$.
Together with Lemma~\ref{lem:sat-base-exactness}, they show that $F_\tau$ is satisfiable exactly when such a successful rule with initial representative $\tau$ exists.

\subsection{Reference encoding and checked proof objects}

The reference encoding $F_{\mathcal H}$ contains the base and strengthening clauses for all nine graphs but does not fix one particular $\tau$.
It instead imposes $r_1,r_2\in\{1,2\}$ and $r_3=1$.
Thus one common rule chooses one of the four representatives and shares every other action across the four cases.
Lemma~\ref{lem:sat-base-exactness} and the soundness of the additional clauses show that $F_{\mathcal H}$ is satisfiable exactly when such a successful rule exists.

For each $F_\tau$, we take its checked unsatisfiable core and remove the unit clauses fixing $\tau$.
We unite the four remaining clause sets, add the three constraints above, and delete duplicates.
The result is the \emph{combined unsatisfiable core}, a subformula of $F_{\mathcal H}$.
The same DRAT refutation has been checked against both this core and the containing reference encoding.

We obtain the \emph{reduced encoding} by removing unused variables, renumbering the remaining variables, and applying core extraction twice.
Its new DRAT proof is checked independently of the original proof.
The reference encoding directly represents the graph-exploration problem.
The combined core and reduced encoding give smaller checked routes to the same contradiction.
Table~\ref{tab:cactus-certificate-size} gives the dimensions of the checked objects.

\begin{table}[htbp]
    \caption{Checked objects for the three-color subcubic-pseudotree lower bound. File sizes are uncompressed bytes.}
    \label{tab:cactus-certificate-size}
    \centering
    \begin{tabular}{@{}lrrr@{}}
        \toprule
        Object                                       & Variables & Clauses   & Bytes         \\
        \midrule
        Reference encoding $F_{\mathcal H}$          & 984,882   & 4,756,124 & 118,444,022   \\
        Combined unsatisfiable core                  & 984,882   & 904,855   & 18,905,669    \\
        Reduced encoding                             & 282,535   & 433,150   & 9,817,630     \\
        Binary DRAT for reference/combined encodings & ---       & ---       & 1,749,789,591 \\
        Binary DRAT for reduced encoding             & ---       & ---       & 1,435,082,233 \\
        \bottomrule
    \end{tabular}
\end{table}

Of the 984,882 variables in $F_{\mathcal H}$, 832,032 are auxiliary variables introduced by the transition-wise binary rank comparators.
The rule, reachability, and rank variables account for the remainder.

\subsection{Reproducing the proof and required trusted software}

The reproducible proof chain runs from $\mathcal H$ in graph6 form through $F_{\mathcal H}$ to a checked DRAT refutation.
The repository preserves the vertex labels, and the verification script checks that its graph6 strings encode the nine graphs in Figure~\ref{fig:cactus-lower-witness}.
For the archived proof, the supplied generator reproduced the $\itup111$ CNF.
A second program replaced the degree-one and degree-two unit clauses with disjunctions that allow either value $1$ or $2$ and retained $r_3=1$.
The resulting formula was $F_{\mathcal H}$.
The verification script compared the generated formula byte for byte with the stored reference CNF used by the checked DRAT proof.
This comparison linked the checked graph and clause inputs to that exact reference CNF.

The repository contains the fixed inputs, generator, checkers, and a script for regenerating the four branch proofs.
It does not contain the large precomputed CNF and DRAT files or the dense-to-original variable map.
No separate precomputed-proof archive is part of this distribution.
The script \texttt{reproduction/regenerate\_and\_verify.sh} checks the graph, trace-clause, and domain-clause inputs, generates $F_\tau$ for each of the four representatives, and compares each CNF with its recorded SHA-256 hash.
It then uses CaDiCaL to produce a DRAT refutation for each branch and checks every refutation with \texttt{drat-trim}.
Together with the initial-tuple reduction, these four checks establish the same lower bound without the archived combined or reduced proof objects.
The repository's README gives the commands and dependencies.
Four parallel branches took about half a day in the reference environment, with about 7\,GB of free disk space and 10\,GB of memory.

As a positive control for the generator, the encoding for the two trees and three cycles in $\mathcal H$, including the merge clauses, is satisfiable.
It remains satisfiable when the rule of Algorithm~\ref{alg:tree-cycle-three} is fixed by unit clauses.
Replacing its action at the degree-one all-zero observation $(0;\{\!\{0\}\!\})$ by a stop makes it unsatisfiable.
Moreover, the base encoding of every single graph in $\mathcal H$ is satisfiable, so no graph in $\mathcal H$ alone witnesses the lower bound, which is a statement about one rule for all nine graphs.

In the reference environment, \texttt{drat-trim} reported \texttt{s VERIFIED} for the reference encoding, combined core, and reduced encoding.
For all three verifications, \texttt{drat-trim} reported zero RAT lemmas.
Every proof addition retained in the checked cores was verified by reverse unit propagation.
The three checks took approximately 164, 132, and 79 minutes, with about 3\,GiB of peak memory in the largest check.
The file \texttt{certificate/SHA256SUMS} records the hashes of these archived CNF, DRAT, and variable-map files.

The trusted base consists of the generator, the independent checkers, the C++ toolchain and runtime, the Node.js runtime, SHA-256, and \texttt{drat-trim}.
It excludes the SAT solver and core-extraction tools, since the generated refutations are checked independently.
The trusted programs comprise about 1{,}700 lines of C++ and 1{,}200 lines of Node.js.
Lemma~\ref{lem:sat-base-exactness} establishes the semantic interpretation of the base encoding.
The source checkers validate the graph and clause inputs used by the generator.
In the regeneration procedure, \texttt{drat-trim} checks the generated $F_\tau$ for all four representatives, and the initial-tuple reduction then gives the lower bound.

\begin{proof}[Proof of Lemma~\ref{lem:four-initial-obstruction}]
    Suppose that a successful rule has one of the four initial representatives.
    By Lemma~\ref{lem:sat-base-exactness}, this rule gives a satisfying assignment of the corresponding base formula $B_\tau$.
    The soundness of the merge clauses and of the independently checked trace and domain clauses ensures that the assignment corresponding to the successful rule also satisfies the additional clauses.
    Hence the reference encoding $F_{\mathcal H}$, which combines the four representatives, is satisfiable.
    However, the reference CNF regenerated from the checked graph and clause inputs is byte-identical to the stored reference CNF, and its DRAT refutation has been checked by \texttt{drat-trim}.
    Thus $F_{\mathcal H}$ is unsatisfiable, a contradiction.
\end{proof}
\end{document}